%% file: main.tex
\documentclass[10pt]{article}
\usepackage[utf8]{inputenc}
\usepackage[margin=1in]{geometry}
\usepackage{amsfonts, amsmath, amssymb, amsthm, bbm, enumitem, dsfont}
\usepackage[usenames,dvipsnames,table]{xcolor}
\usepackage{preamble}

\usepackage{palatino}

\newcommand{\fr}{\frac}
\newcommand{\lt}{\left}
\newcommand{\rt}{\right}

\newcommand*{\defeq}{\stackrel{\text{def}}{=}}

\newcommand{\ind}{\mathbbm{1}}

\newcommand{\bbP}{\mathbb{P}}
\newcommand{\bbR}{\mathbb{R}}

\newcommand{\cA}{\mathcal{A}}
\newcommand{\cE}{\mathcal{E}}
\newcommand{\cM}{\mathcal{M}}
\newcommand{\cN}{\mathcal{N}}

\newcommand{\cT}{\mathcal{T}}
\newcommand{\cZ}{\mathcal{Z}}

\newcommand{\bx}{{\boldsymbol x}}

\newcommand{\diag}{\mathrm{diag}}

\newcommand{\good}{\mathsf{good}}
\newcommand{\Split}{\mathsf{Split}}

\newcommand{\tr}{\mathsf{tr}}
\newcommand{\Sp}{\mathsf{Sp}}

\newcommand{\Leb}{{\mathsf{Leb}}}

\newcommand{\eps}{\varepsilon}

\newcommand{\GUE}{\mathrm{GUE}}
\newcommand{\GUEc}{\mathrm{GUE}^*}

\newcommand{\de}{{\mathrm{d}}}

\title{An optimal tradeoff between entanglement and copy complexity for state tomography}
\author{
    Sitan Chen \thanks{Harvard SEAS, \url{sitan@seas.harvard.edu}}
    \and
    Jerry Li \thanks{Microsoft Research, \url{jerrl@microsoft.com}}
    \and
    Allen Liu \thanks{MIT, \url{cliu568@mit.edu}.  This work was supported in part by an NSF Graduate Research Fellowship and a Fannie and John Hertz Foundation Fellowship}
}
\date{}

\begin{document}

\maketitle

\begin{abstract}
    There has been significant interest in understanding how  practical constraints on contemporary quantum devices impact the complexity of quantum learning. For the classic question of tomography, recent work~\cite{chen2023does} tightly characterized the copy complexity for any protocol that can only measure one copy of the unknown state at a time, showing it is polynomially worse than if one can make fully-entangled measurements. While we now have a fairly complete picture of the rates for such tasks in the near-term and fault-tolerant regimes, it remains poorly understood what the landscape in between these extremes looks like, and in particular how to gracefully scale up our protocols as we transition away from NISQ.

    In this work, we study tomography in the natural setting where one can make measurements of $t$ copies at a time. For sufficiently small $\epsilon$, we show that for any $t \le d^2$, $\widetilde{\Theta}(\frac{d^3}{\sqrt{t}\epsilon^2})$ copies are necessary and sufficient to learn an unknown $d$-dimensional state $\rho$ to trace distance $\epsilon$. This gives a smooth and optimal interpolation between the known rates for single-copy measurements and fully-entangled measurements. 
    
    To our knowledge, this is the first smooth entanglement-copy tradeoff known for any quantum learning task, and for tomography, no intermediate point on this curve was known, even at $t = 2$. An important obstacle is that unlike the optimal single-copy protocol~\cite{kueng2017low,guctua2020fast}, the optimal fully-entangled protocol \cite{haah2016sample,o2016efficient} is inherently a biased estimator. This bias precludes naive batching approaches for interpolating between the two protocols. Instead, we devise a novel two-stage procedure that uses Keyl's algorithm \cite{keyl2006quantum} to refine a crude estimate for $\rho$ based on single-copy measurements. A key insight is to use Schur-Weyl sampling not to estimate the spectrum of $\rho$, but to estimate the \emph{deviation} of $\rho$ from the maximally mixed state.  When $\rho$ is far from the maximally mixed state, we devise a novel quantum splitting procedure that reduces to the case where $\rho$ is close to maximally mixed.
\end{abstract}

\section{Introduction}

As the computational resources for near-term quantum devices continue to grow, so too does their potential to help us analyze quantum experimental data and learn about the physical universe. It is timely not only to understand the fundamental limitations contemporary devices impose for such tasks relative to fault-tolerant quantum computation, but also to map out avenues for \emph{gracefully scaling} up our near-term algorithms as the platforms on which we run these algorithms mature. 


Here we examine this challenge for arguably the most fundamental problem in quantum data analysis: state tomography. Recall that in this problem, one is given $n$ copies of a $d$-dimensional mixed state $\rho$, and the goal is to estimate its density matrix to error $\epsilon$, e.g. in trace distance, by measuring copies of $\rho$. The standard figure of merit here is \emph{copy complexity}: how small does $n$ have to be, in terms of $d$ and $1/\epsilon$?

Given fault-tolerant quantum computers,~\cite{haah2016sample} and~\cite{o2016efficient} settled the optimal copy complexity of this problem, showing that $n = \Theta(d^2/\epsilon^2)$ copies are necessary and sufficient. The protocols they gave came with an important caveat however: they require fully coherent measurements across the joint state $\rho^{\otimes n}$. In conventional experimental setups, this renders the algorithms impractical, as typically we do not have the ability to simultaneously prepare and store so many copies of our quantum state. Motivated by contemporary device constraints, there has been a great deal of recent interest in understanding the other extreme, namely what one can do with \emph{incoherent} (a.k.a. \emph{independent} or \emph{unentangled}) measurements, where the algorithm measures a single copy of $\rho$ at a time.
These algorithms are much more practical to run---now, the experimenter merely needs to prepare a single copy of the state at a time, interact with it once, and repeat.
Recent work of~\cite{chen2023does} determined the optimal copy complexity in this setting, showing that $n = \Theta(d^3/\epsilon^2)$ copies are necessary and sufficient and thus that incoherent measurements are provably weaker than coherent measurements.

That said, the incoherent setting increasingly appears to be an overly restrictive proxy for current platforms, and realistic experimental settings.
For instance, an experimenter could simultaneously prepare $t$ copies of the state by simply replicating her experimental setup $t$ times, and some existing devices can already store multiple (albeit few) copies of a quantum state at a time~\cite{arute2019quantum,islam2015measuring,linke2018measuring}.
In many settings, this is already very powerful.
For instance, recent experimental demonstration of exponential advantage for estimating Pauli observables using Bell measurements naturally operate on two copies of the input state at a time~\cite{huang2022quantum}, and cannot be done with incoherent measurements. In fact, the incoherent setting precludes even basic quantum learning primitives like the SWAP test.

Yet remarkably, prior to the present work, it was not even known how to rule out the possibility of achieving $\Theta(d^2/\epsilon^2)$ copy complexity just using measurements of $t = 2$ copies of $\rho$ at a time! This begs the natural question:

\begin{center}
	\emph{When are $t$-entangled measurements asymptotically stronger than incoherent measurements for state tomography?}
\end{center}

Apart from the practical motivation, understanding $t$-entangled measurements also poses an important conceptual challenge in theory, namely the fundamental incompatability between existing algorithms in the fully entangled and incoherent settings.

The fully entangled algorithms of~\cite{haah2016sample} and~\cite{o2016efficient} crucially rely on estimates derived from weak Schur sampling, which is grounded in the representation theory of the symmetric and general linear groups.
Indeed, if one is only allowed a single entangled measurement, then any optimal measurement provably must first go through weak Schur sampling \cite{Keyl2001estimating}.
However, at the same time, this approach seems inherently tied to such ``one-shot'' algorithms. For example, while one can show that sampling from the Schur-Weyl distribution yields an estimate for the eigenvalues of $\rho$ which is close in expectation, it is actually a \emph{biased} estimator for the spectrum (indeed, its expectation always majorizes the spectrum, see Lemma 4.1.2 in~\cite{wright2016learn}). This bias decreases with $t$, and unless $t = \Omega (d^2/\eps^2)$, it is unclear if weak Schur sampling on $t$ copies gives any useful information, let alone if one can somehow combine the outcome of several different trials of weak Schur sampling in a meaningful fashion.

This is in contrast to algorithms for incoherent measurements, which rely on many unbiased (but less informative) measurements of the underlying state, which can then be averaged to obtain a good overall estimate.
Unfortunately, this unbiasedness seems specific to incoherent measurements.
Ultimately, neither approach by itself seems capable of yielding a nontrivial result for general $t$-entangled measurements. This calls for a new algorithmic framework that can synthesize the two techniques.

\paragraph{Our results} In this paper, we give a tight characterization of the copy complexity of state tomography with $t$-entangled measurements, for a large range of $t$.
More specifically, we show the following pair of theorems.
Here, and throughout the paper, we let $\widetilde{O}, \widetilde{\Omega}$, and $\widetilde{\Theta}$ hide polylogarithmic factors in $d, \eps$, and we let $\norm{\cdot}_1$ denote the trace norm of a $d \times d$ matrix.
\begin{theorem}[informal upper bound, see Theorem~\ref{thm:learn-general}]\label{thm:upper-informal}
    Let $t \leq \min (d^2, (\sqrt{d} / \eps)^c)$, for some (sufficiently small) universal constant $c$. There is an algorithm that uses 
    \[
    n = \widetilde{O} \left( \frac{d^3}{\sqrt{t} \eps^2}\right) \; .
    \]
    total copies of $\rho$ by taking $n / t$ separate $t$-entangled measurements, and outputs $\hat{\rho}$ so that $\norm{\rho - \hat{\rho}}_1 \leq \eps$ with high probability.
\end{theorem}
\begin{theorem}[informal lower bound, see Theorem~\ref{thm:tomo_lbd}]\label{thm:lower-informal}
    Let $t \leq 1 / \eps^c$ for some (sufficiently small) universal constant $c$. Any algorithm that uses $n$ total copies of $\rho$ by taking taking $n / t$ separate (but possible adaptively chosen) $t$-entangled measurements must require 
    \[
    n = \Omega \left( \frac{d^3}{\sqrt{t} \eps^2}\right) 
    \]
    to succeed at state tomography, with probability $\geq 0.01$.
\end{theorem}
Together, these theorems imply that the copy complexity of tomography with $t$-entangled measurements is, up to logarithmic factors, exactly $\widetilde{\Theta} \left( \tfrac{d^3}{\sqrt{t} \eps^2} \right)$, for all $t \leq 1 / \eps^c$. 
In other words, for this regime of $t$, the sample complexity smoothly improves with $t$.
To our knowledge, this is the first nontrivial setting in which the copy complexity has been sharply characterized to smoothly depend on the amount of entanglement.
Prior works on the power of $t$-entangled measurements either did not obtain sharp tradeoffs, or studied settings where simply taking $t$ to be a large enough constant sufficed to obtain optimal copy complexity (e.g. Bell sampling \cite{chen2022exponential, montanaro2017learning}).

We pause here to make a few remarks about our result.
First, as mentioned above, the most practically pertinent regime of $t$ is the regime where $t$ is small but $> 1$.
Our result fully characterizes this setting, so long as $\eps = o(1)$, and obtains a nontrivial partial tradeoff even for constant $\epsilon$.
We believe that the restriction that $t \leq 1 / \eps^c$ is ultimately an artifact of our techniques, and we conjecture that $\widetilde{\Theta} \left( \tfrac{d^3}{\sqrt{t} \eps^2} \right)$ is the right rate for all $t \leq O(d^2)$. We leave resolving this question as an interesting future direction.

Second, we note that our result implies that when $\eps$ is sufficiently small, entanglement only helps up until a certain point.
When $\eps \ll 1/d$, we can take $t = \Theta(d^2)$, and our algorithm achieves a rate of $n = \widetilde{O} (\tfrac{d^2}{\eps^2})$, which matches (up to logarithmic factors) the lower bound against general quantum algorithms with any degree of entanglement.
In other words, we can match the rate of the fully coherent algorithm, with asymptotically less coherence.
To our knowledge, this is the first instance for a natural learning problem where a super-constant, but still partial, amount of entanglement is the minimal amount of entanglement required to obtain the statistically optimal rate.

\paragraph{Our techniques.} As previously mentioned, we need a number of novel ideas to overcome the difficulties with prior techniques for the $t$-entangled setting.
In this section, we describe some of these conceptual contributions at a high level.
For a more in-depth discussion of our techniques, we refer the reader to Section~\ref{sec:overview}.

A key idea, which is crucial to both our upper and lower bounds, and which to our knowledge is novel in the literature, is something we call \emph{linearization}.
For any state $\rho$, we can write it as $\rho = \frac{I_d}{d} + E$, where $E$ is some traceless Hermitian matrix.
Then, for any integer $t$, we observe that $\rho^{\otimes t}$ can be approximated as follows
\begin{equation}
\label{eq:linearization}
\rho^{\otimes t} \approx \left( \frac{1}{d} I_d \right)^{\otimes t} + \sum_{\mathrm{sym}} E\otimes \left( \frac{1}{d} I_d \right)^{\otimes t-1} \; ,
\end{equation}
where here and throughout the paper, we let $\sum_{\mathrm{sym}}$ denote the symmetric sum, i.e.
\[
\sum_{\mathrm{sym}} A_1 \otimes \ldots \otimes A_t = \sum_{\pi \in S_t} A_{\pi(1)} \otimes \ldots \otimes A_{\pi(t)} \; .
\]
We call the expression on the RHS of~\eqref{eq:linearization} the \emph{$t$-linearization} of $\rho$, as it is the linear term (in $E$) of the expansion of $\rho^{\otimes t}$.
Note that while the linearization of the state is not necessarily itself a mixed state (as it need not be PSD), we can still perform formal calculations with it.
For our purposes, the linearization of $\rho$ has a number of crucial properties.
While estimators based on weak Schur sampling do not yield an unbiased estimate of $\rho$ when applied to $\rho^{\otimes t}$, they do yield an unbiased estimate of $E$ (up to a known correction term) when ``applied'' to the $t$-linearization of $\rho$.
Therefore, as long as the linearization of $\rho$ is a good approximation for $\rho^{\otimes t}$, we may average the result of many independent trials of such estimators to obtain a better estimate for $E$, and consequently, $\rho$.
For the lower bound, it turns out that the amount of information gained by the algorithm from some measurement outcome (more formally, the likelihood ratio of the posterior distribution for the algorithm starting from two different states, see Lemma~\ref{lem:pointwise-ratio-bound}) is controlled by the Frobenius norm of some linear transformation of the corresponding POVM element.

Another important and novel algorithmic idea is that of \emph{quantum splitting}.
It turns out that our techniques work best for states which are relatively spread-out, i.e. states with spectral norm at most $O(1/d)$.
To extend our techniques to general mixed states, we give a novel procedure which reversibly transforms any state into another state of marginally larger dimension that is spectrally bounded.
One can view this as a quantum generalization of the splitting procedure introduced in~\cite{diakonikolas2016new} for classical distribution learning and testing.
There, as well, the goal was to take an arbitrary distribution over $[d]$ elements and produce a ``well-spread'' distribution over a slightly larger domain.
However, more care must be taken in the quantum setting, as our splitting procedure must be done obliviously of the (unknown) eigenvectors of the true state $\rho$, and so any such procedure will non-trivially alter the spectral properties of $\rho$.

\subsection{Related work}

This work comes in the wake of a flurry of recent works characterizing the effect that near-term constraints like noise and limited quantum memory have on statistical rates for quantum learning tasks. Many of these have focused on proving lower bounds for protocols that can only make incoherent measurements, e.g.~\cite{bubeck2020entanglement,chen2021toward,chen2022tight,chen2021hierarchy,chen2022exponential,aharonov2022quantum,fawzi2023quantum,chen2023does}. Among these works, the two most relevant ones are~\cite{chen2023does} and~\cite{chen2021hierarchy}. 

In~\cite{chen2023does}, the authors prove an optimal lower bound on the copy complexity of state tomography with incoherent measurements, showing that the copy complexity upper bound of $O(\frac{d^3}{\epsilon^2})$ originally obtained by~\cite{kueng2017low} is tight. The proof of the lower bound in Theorem~\ref{thm:lower-informal} builds upon the general proof strategy developed in that work. Roughly speaking, in lieu of standard mutual information-based proofs of minimax lower bounds, this framework involves showing that for some prior over input states, if one observes a typical sequence of measurement outcomes resulting from any protocol that makes insufficiently many incoherent measurements, the conditional distribution over the input state places negligible mass on the true state. As we explain in Section~\ref{sec:overview}, there are a number of obstacles that we must overcome in order to adapt this strategy to the $t$-entangled setting.

In~\cite{chen2021hierarchy}, the authors are motivated by the same general question as the present work: are there tasks for which being able to perform $t$-entangled measurements for large $t$ is strictly more powerful than for small $t$? Instead of tomography, they studied a certain hypothesis testing task based on distinguishing approximate state $k$-designs on $n$ qubits from maximally mixed. They show that when $t < k$, any protocol must use exponentially many copies, but with fully-entangled measurements, $\poly(n,k,1/\epsilon)$ copies suffice because one can simply run quantum hypothesis selection~\cite{buadescu2021improved}. Unlike the present work, they do not manage to give a full tradeoff. Additionally, they prove some partial results for using $t$-entangled measurements for mixedness testing~\cite{o2015quantum,bubeck2020entanglement,chen2022tight}, where the goal is to distinguish the maximally mixed state from states that are $\epsilon$-far in trace distance from maximally mixed. Specifically, when $t = o(\log(d))/\epsilon^2$, the polynomial dependence on $d$ in their lower bound matches that of the lower bound of~\cite{bubeck2020entanglement} for $t = 1$. Although that lower bound is suboptimal even for $t = 1$, and although~\cite{chen2021hierarchy} did not achieve a full tradeoff for mixedness testing, we remark that it hints at an interesting phenomenon. 
Unlike in tomography, where copy complexity smoothly decreases as $t$ increases, for mixedness testing, the number of copies $t$ must exceed some \emph{$\epsilon$-dependent threshold} before one can improve upon the single-copy rate.

Lastly, we remark that our analysis of the upper bound makes crucial use of certain representation-theoretic estimates proved by~\cite{o2016efficient}. While these estimates were originally used to analyze Keyl's estimator in the fully-entangled setting, in our analysis we leverage them in a novel way to prove guarantees in the partially-entangled setting (see Section~\ref{sec:overview} for details).

\subsection{Discussion and open problems}

In this work we gave, to our knowledge, the first smooth tradeoff between the copy complexity of a quantum learning task and the number of copies that one can measure at once. When the target error $\epsilon$ is sufficiently small, we show that up to logarithmic factors, the copy complexity achieved by our protocol achieves the optimal tradeoff between the best single-copy and fully-entangled protocols for the full range of possible $t$. Below we mention a number of interesting open directions to explore:

\paragraph{Larger $\epsilon$.} The most obvious axis along which our results could be improved is that we do not achieve the full tradeoff when $\epsilon$ is somewhat large. Indeed, our upper and lower bounds require respectively that $t\le (\sqrt{d}/\epsilon)^c$ and $t \le (1/\epsilon)^{c'}$ for constants $c,c'$, so e.g. for $\epsilon$ which is a small constant, we can only show our $t$-entangled protocol is optimal for $t$ at most some constant. In some sense this is unavoidable with our current techniques as we crucially exploit the fact that for any state $\rho = I_d/d + E$, its tensor power $\rho^{\otimes t}$ is well-approximated by its linearization, provided $E$ is sufficiently small relative to $t$. A larger value of $\epsilon$ corresponds to a larger deviation $E$ relative to $t$ in our arguments, and at a certain point this linear approximation breaks down. It is interesting to see whether one can leverage higher-order terms in the expansion of $\rho^{\otimes t}$ to handle larger $\epsilon$.

\paragraph{Rank-dependent rates.} For fully-entangled measurements, it is known~\cite{o2016efficient,haah2016sample} how to achieve sample complexity $O(rd/\epsilon^2)$ when $\rho$ has rank $r$, e.g. using a ``truncated'' version of Keyl's algorithm. Even when $t = 1$, it is open what the optimal sample complexity is for rank-$r$ state tomography; conjecturally, the answer is the rate given by the best known upper bound of $O(dr^2/\epsilon^2)$~\cite{kueng2017low,guctua2020fast}. It is an interesting question whether the protocol in the present work can be adapted to give a smooth tradeoff between the best-known single-copy and fully-entangled algorithms.

\paragraph{Quantum memory lower bounds.} While the model of $t$-entangled measurements is a natural setting for understanding the power of quantum algorithms that have some nontrivial amount of quantum memory, it is not the most general model for such algorithms. More generally, one could consider a setting where there are $tn$ qubits, where $n \defeq \log_2 d$, on which the algorithm can perform arbitrary quantum computation interspersed with calls to a state oracle given by the channel
\begin{equation}
    \sigma \mapsto \rho\otimes \Tr_{1:n}(\sigma)\,,
\end{equation}
where $\sigma$ is an arbitrary $tn$-qubit density matrix and $\Tr_{1:n}(\cdot)$ denotes partial trace over the first $n$ qubits. To date, the only known lower bounds for learning tasks in this general setting are in the context of Pauli shadow tomography for both states~\cite{chen2022exponential} and processes~\cite{chen2022quantum,caro2022learning,chen2023futility, chen2023tight}, as well as learning noisy parity~\cite{liu2023memory}, and the techniques in these works break down as one approaches $t = 2$. It is an outstanding open question whether one can obtain nontrivial quantum memory lower bounds even for $2n$ qubits of quantum memory, not just for state tomography but for any natural quantum learning task.

\input{tech-overview}

\section{Representation theory basics}\label{sec:rep-theory}
We will introduce some notation and facts from representation theory for analyzing entangled quantum measurements.  Our exposition closely follows \cite{o2016efficient,o2017efficient}.  For a more detailed explanation of the elementary representation theory results, see e.g. \cite{wright2016learn}.   We use $GL_d$ to denote the general linear group in $\C^{d \times d}$ and $U_d$ to denote the unitary group in $\C^{d \times d}$.  Next, we introduce some notation for partitions.



\begin{definition}
Given a positive integer $n$, a partition of $n$ into $d$ parts is a list of positive integers $\lambda = (\lambda_1, \dots , \lambda_d)$ with $\lambda_1 \geq \dots \geq \lambda_d$ such that $\lambda_1 + \dots + \lambda_d = n$.  We write $\lambda \vdash n$ to denote that $\lambda$ is a partition of $n$.  We use $|\lambda|$ to denote the total number of elements in $\lambda$ and $\ell(\lambda)$ to denote the number of nonzero parts in $\lambda$, e.g. in the previous example $|\lambda| = n$ and $\ell(\lambda) = d$.
\end{definition}

\begin{definition}
Given two partitions $\lambda, \lambda' \vdash n$ where $\lambda = (\lambda_1, \dots , \lambda_d)$ and $\lambda' = (\lambda_1', \dots , \lambda_{d'}')$, we say $\lambda$ majorizes $\lambda'$ if 
\[
\lambda_1 + \dots + \lambda_i \geq \lambda_1' + \dots + \lambda_i'
\]
for all $i$ where $\lambda_j, \lambda_j'$ are defined to be $0$ whenever $j$ exceeds the number of parts in the partition.  If $\lambda$ majorizes $\lambda'$, we write $\lambda \succeq \lambda'$.
\end{definition}

\begin{fact}[\cite{romik2015surprising}]\label{fact:num-partitions}
For any $n$, the number of distinct partitions $\lambda \vdash n$ is at most $2^{3\sqrt{n}}$.
\end{fact}

\begin{definition}\label{def:count}
Given a partition $\lambda \vdash n$, let $f_i(\lambda)$ be the number of parts of $\lambda$ that are equal to $i$ for each integer $i$.  We define $\mathsf{count}(\lambda) = f_1(\lambda)! \cdots f_n(\lambda)!$.  
\end{definition}

\begin{definition}
Given a partition $\lambda \vdash n$, we let $\lambda^T \vdash n$ denote its transpose i.e. $\lambda^T_1$ the number of parts of $\lambda$ that are at least $1$, $\lambda^T_2$ is the number of parts of $\lambda$ that are at least $2$ and so on.    
\end{definition}

\begin{definition}\label{def:young-tableaux}[Young Tableaux]
We have the following standard definitions:
\begin{itemize}
\item Given a partition $\lambda \vdash n$, a Young diagram of shape $\lambda$ is a left-justified set of boxes arranged in rows, with $\lambda_i$ boxes in the $i$th row from the top.
\item  A standard Young tableaux (SYT) $T$ of shape $\lambda$ is a Young diagram of shape $\lambda$ where each box is filled with some integer in $[n]$ such that the rows are strictly increasing from left to right and the columns are strictly increasing from top to bottom.
\item A semistandard Young tableaux (SSYT) $T$ of shape $\lambda$ is a Young diagram of shape $\lambda$ where each box is filled with some integer in $[d]$ for some $d$ and the rows are weakly increasing from left to right and the columns are strictly increasing from top to bottom.
\end{itemize}
\end{definition}

Now we review the correspondence between Young tableaux and representations of the symmetric and general linear groups.

\begin{definition}
We say a representation $\mu$ of $GL_d$ over a complex vector space $\C^m$ is a polynomial representation if for any $U \in \C^{d \times d}$, $\mu(U) \in \C^{m \times m}$ is a polynomial in the entries of $U$.
\end{definition}

\begin{fact}[\cite{sagan2013symmetric}]
The irreducible representations of the symmetric group $S_n$ are exactly indexed by the partitions $\lambda \vdash n$ and have dimensions $\dim(\lambda)$ equal to the number of standard Young tableaux of shape $\lambda$.  We denote the corresponding vector space $\Sp_{\lambda}$.
\end{fact}

\begin{fact}[\cite{goodman2009symmetry}]
For each $\lambda \vdash n$, there is a (unique) irreducible polynomial representation of $GL_d$ corresponding to $\lambda$. We denote the corresponding map and vector space $(\pi_{\lambda}, V_{\lambda}^d)$.  The dimension $\dim(V_{\lambda}^d)$ is equal to the number of semistandard Young tableaux of shape $\lambda$ with entries in $[d]$.  This representation, restricted to $U_d$ is also an irreducible representation.
\end{fact}

\begin{theorem}[Schur-Weyl Duality \cite{goodman2009symmetry}]\label{thm:schur-weyl}
Consider the representation of $S_n \times GL_d$ on $(\C^{d})^{\otimes n}$ where the action of the permutation $\pi \in S_n$ permutes the different copies of $\C^{d}$ and the action of $U \in GL_d$ is applied independently to each copy.  This representation can be decomposed as a direct sum 
\[
(\C^{d})^{\otimes n} = \bigoplus_{\substack{\lambda \vdash n \\ \ell(\lambda) \leq d  }} \Sp_{\lambda} \otimes V_{\lambda}^d \,.
\]
\end{theorem}

\begin{definition}[Schur Subspace]
We call $\Sp_{\lambda} \otimes V_{\lambda}^d$ the $\lambda$-Schur subspace.  Given integers $n,d$ and $\lambda \vdash n$, we define $\Pi_{\lambda}^d: (\C^{d})^{\otimes n} \rightarrow  \Sp_{\lambda} \otimes V_{\lambda}^d $ to project onto the $\lambda$-Schur subspace.  
\end{definition}

\begin{theorem}[Gelfand-Tsetlin Basis \cite{goodman2009symmetry}]\label{thm:gf-basis}
Let $n,d$ be positive integers.  For each partition $\lambda \vdash n$ where $\lambda$ has at most $d$ parts, there is a basis $v_1, \dots , v_m$ of $V_{\lambda}^d$ with $m = \dim(V_{\lambda}^d)$ such that for any matrix $D_{\alpha} = \diag(\alpha_1, \dots , \alpha_d)$, we have $v_i^\dagger \pi_{\lambda}(D_{\alpha})v_i = \alpha^{f^{(i)}}$ for all $i$ where $f^{(i)}$ are each $d$-tuples that give the frequencies of $1,2, \dots , d$ in each of the different semi-standard tableaux of shape $\lambda$.
\end{theorem}

We now present a few consequences of Theorem~\ref{thm:schur-weyl} and Theorem~\ref{thm:gf-basis} that will be used in our learning primitives.  Roughly, we give a more explicit representation of the Gelfand-Tsetlin basis when embedded in the representation of $(\C^{d})^{\otimes n}$.

\begin{definition}
For integers $n,d$ and a $d$-tuple $f_1, \dots , f_d$ such that $f_1 + \dots + f_d = n$, we define $G_{f_1, \dots , f_d}[n \rightarrow d ]$ to be the set of all functions $g:[n] \rightarrow [d]$  where the multiset $\{g(1), \dots , g(n) \}$ has $1$ with frequency $f_1$, $2$ with frequency $f_2$ and so on. 
\end{definition}

\begin{lemma}\label{lem:weight-vector}
Let $f = (f_1, \dots , f_d)$ where $f_1, \dots ,  f_d \geq 0$ are integers and  $f_1 + \dots + f_d = n$.  Let $v_1, \dots , v_d$ be orthogonal vectors in $\C^d$.  If $\lambda$ is a partition that majorizes $(f_1, \dots , f_d)$ (after sorting in decreasing order) then there are weights $\{ w_{g}\}_{g \in G_{f_1, \dots , f_d}[n \rightarrow d ]}$ such that the vector
\[
v_f = \sum_{g \in G_{f_1, \dots , f_d}[n \rightarrow d ]} w_{g} (v_{g(1)} \otimes \cdots \otimes v_{g(n)}) 
\]
is in the Schur subspace $\Sp_{\lambda} \otimes V_{\lambda}^d$.  Furthermore, we can choose $\dim(\lambda)$ orthogonal unit vectors $v_{f,1}, \dots , v_{f,\dim(\lambda)}$ of the above form the weights so that 
\[
\E_{U}\left[U^{\otimes n} \left(\sum_{j = 1}^{\dim(\lambda)} v_{f,j} v_{f,j}^\dagger \right) (U^{\dagger})^{\otimes n}  \right] = \frac{1}{\dim(V_{\lambda}^d)} \Pi_{\lambda}^d
\]
where the expectation is over Haar random unitaries $U$.
\end{lemma}
\begin{proof}
First, WLOG $v_1, \dots , v_d$ are the standard basis $e_1, \dots , e_d$.  This is because by Theorem~\ref{thm:schur-weyl}, we can always rotate all of the $n$ copies simultaneously by the same unitary while staying within the same Schur subspace.  Now imagine decomposing $\C^{d^n}$ into the subspaces given by Theorem~\ref{thm:schur-weyl}.  Take one of the copies of $V_{\lambda}^d$.  Note that since $\lambda$ majorizes $(f_1, \dots , f_d)$, there is a semi-standard tableaux of shape $\lambda$ where the numbers $1,2, \dots , d$ occur exactly $f_1, \dots , f_d$ times respectively.  Thus, we can apply Theorem~\ref{thm:gf-basis} to find the vector $v_f$ in it such that 
\[
v_f^\dagger D_{\alpha}^{\otimes n} v_f = \alpha^f
\]
where $D_{\alpha} = \diag(\alpha_1, \dots , \alpha_d)$.  Note that both sides of the above are polynomials in the $\alpha_1, \dots , \alpha_d$ and it holds for all values, so it must actually hold as a polynomial identity.  Thus, $v_f$ must be contained in the eigenspace of $D_{\alpha}^{\otimes n}$ with eigenvalue $\alpha^f$ which is exactly the subsapce spanned by $e_{g(1)} \otimes \dots \otimes e_{g(n)}$ as $g$ ranges over all of the functions in $G_{f_1, \dots , f_d}[n \rightarrow d ]$.  This immediately implies the first statement. 

For the second statement, the fact that $V_{\lambda}^d$ is an irreducible representation of $U_d$ implies that 
\[
\E_{U}\left[U^{\otimes n}  v_{f} v_{f}^\dagger  (U^{\dagger})^{\otimes n}  \right]
\]
is equal to $\frac{1}{\dim(V_{\lambda}^d)}$ times a projection onto the copy of $V_{\lambda}^d$ that $v_f$ is in.  Thus, we can simply use the same construction and pick $v_{f,1}, \dots , v_{f, \dim(\lambda)}$, one from each of the copies of $V_{\lambda}^d$ in $\C^{d^n}$ and then we get the desired statement.
\end{proof}

In light of Lemma~\ref{lem:weight-vector}, we make the following definition.
\begin{definition}\label{def:weight-vector}
Given integers $n,d$ and a partition $\lambda \vdash n$,  we define the vectors $v_{\lambda,1}^d, \dots , v_{\lambda, \dim(\lambda)}^d \in \C^{d^n}$ to be the vectors constructed in Lemma~\ref{lem:weight-vector} where we choose $v_1, \dots , v_d$ to be the standard basis and $(f_1, \dots , f_d) = (\lambda_1, \dots , \lambda_d)$ (when the vectors are not unique, we pick arbitrary ones, the choice will not matter when we use this later on).  We define
\[
M_{\lambda}^d = \sum_{j = 1}^{\dim(\lambda)} v_{\lambda, j}^d (v_{\lambda, j}^d)^\dagger \,.
\]
We may drop the $d$ in the superscript and simply write $M_{\lambda}$ when $d$ is clear from context.
\end{definition}

We also have the following lemma that provides a type of converse to Lemma~\ref{lem:weight-vector}.
\begin{lemma}\label{lem:weight-vector2}
Let $g:[n] \rightarrow [d]$ be a function in $G_{f_1, \dots , f_d}[n \rightarrow d]$.  Then for any partition $\lambda \vdash n$ that does not majorize $(f_1, \dots , f_d)$ (after sorting in decreasing order), we have
\[
\Pi_{\lambda}^d(v_{g(1)} \otimes  \dots \otimes v_{g(n)}) = 0 \,.
\]
\end{lemma}
\begin{proof}
As before, due to Theorem~\ref{thm:schur-weyl}, WLOG $v_1, \dots , v_d$ are the standard basis $e_1, \dots , e_d$.  Take any of the copies of $V_{\lambda}^d$.  Note that by assumption, there does not exist any semi-standard tableaux of shape $\lambda$ where $1,2, \dots , d$ occur with frequencies $f_1, \dots , f_d$ respectively.  Now consider a basis $v_{f^{(1)}}, \dots , v_{f^{(m)}}$  for $V_{\lambda}^d$ as given by Theorem~\ref{thm:gf-basis}. As in Lemma~\ref{lem:weight-vector}, we have that each vector $v_{f^{(i)}}$ must be contained in the span of vectors $v_{g'(1)} \otimes \dots \otimes v_{g'(n)}$ as $g'$ ranges over all functions in $G_{f^{(i)}_1 , \dots , f^{(i)}_d}[n \rightarrow d]$.  However, as no semi-standard tableaux of shape $\lambda$ has  $1,2, \dots , d$ occur with frequencies $f_1, \dots , f_d$, these subsapces are all orthogonal to $v_{g(1)} \otimes  \dots \otimes v_{g(n)}$ and thus $v_{f^{(i)}}$ is orthogonal as well.  Since the $v_{f^{(i)}}$ form a basis for (one copy of) $V_{\lambda}^d$ and we can repeat the same argument for all of the other copies, $v_{g(1)} \otimes  \dots \otimes v_{g(n)}$ is actually orthogonal to the entire $\lambda$-Schur subspace and we are done. 
\end{proof}

Equipped with the above constructions, we can define the following POVMs.

\begin{definition}[Weak Schur Sampling]
We use the term weak Schur sampling to refer to the POVM on $\C^{d^n \times d^n}$ with elements given by $\Pi_\lambda^d$ for $\lambda$ ranging over all partitions of $n$ into at most $d$ parts.
\end{definition}

Our algorithm will make use of a POVM introduced by Keyl in \cite{keyl2006quantum}.  However, the actual estimator we construct and its analysis will (necessarily) be very different from previous works.

\begin{definition}[Keyl's POVM \cite{keyl2006quantum}]\label{def:keyl-povm}
We define the following POVM on $\C^{d^n \times d^n}$:  first perform weak Schur sampling to obtain $\lambda \vdash n$.  Then within the subspace $\Sp_{\lambda} \otimes V_{\lambda}^d$, measure according to
\[
\dim(V_{\lambda}^d)  \{ U^{\otimes n} M_{\lambda} (U^{\dagger})^{\otimes n}) \}_{U}
\]
where $U$ ranges over Haar random unitaries.  Note that the outcome of the measurement consists of a partition $\lambda \vdash n$ and a unitary $U \in \C^{d \times d}$.
\end{definition}
\begin{remark}
The fact that this is a valid POVM follows from Lemma~\ref{lem:weight-vector}.    
\end{remark}

\subsection{Schur sampling and the Schur-Weyl distribution}

Next, we introduce a few standard tools for analyzing entangled measurements involving Schur sampling and related POVMs.

\begin{definition}
For a $d$-tuple of variables $x = (x_1, \dots , x_d)$ and integers $\lambda = (\lambda_1, \dots , \lambda_d)$, we use $\sum_{\mathrm{sym}} x^{\lambda}$ to denote the symmetric polynomial that is a sum over all distinct monomials of the form $x_1^{\lambda_{\pi(1)}} \cdots x_d^{\lambda_{\pi(d)}}$ for permutations $\pi$ on $[d]$.
\end{definition}

\begin{definition}[Schur Polynomial]
Given integers $n,d$ and a partition $\lambda \vdash n$, we define the $d$-variate polynomial $s_{\lambda}(x_1, \dots , x_d)$ as 
\[
s_{\lambda}(x_1, \dots , x_d) = \sum_{\substack{\text{SSYT } T: \\ T \ \text{has shape $\lambda$}}} x^T 
\]
where $x^T = x_1^{f_1} \cdots x_d^{f_d}$ where $f_1, \dots , f_d$ are the frequencies of $1,2, \dots , d$ in the tableaux $T$.
\end{definition}

\begin{fact}[\cite{goodman2009symmetry}]
We have the following properties:
\begin{itemize}
\item The Schur polynomial $s_{\lambda}(x_1, \dots , x_d)$   is symmetric in the variables $x_1, \dots , x_d$
\item The Schur polynomial is equal to the character of the representation $\pi_{\lambda}$ i.e. for a matrix $M \in \C^{d \times d}$ with eigenvalues $\alpha_1, \dots , \alpha_d$, $\tr(\pi_{\lambda}(M)) = s_{\lambda}(\alpha_1, \dots , \alpha_d)$.
\item We have for any $\alpha_1, \dots , \alpha_d$, the identity
\[
\sum_{\lambda \vdash n} \dim(\lambda) s_{\lambda}(\alpha_1, \dots , \alpha_d) = (\alpha_1 + \dots + \alpha_d)^n
\]
\end{itemize}
\end{fact}

\begin{definition}[Schur Weyl Distribution]
Given integers $n,d$ and a tuple $(\alpha_1, \dots , \alpha_d)$ with $\alpha_i \geq 0$ and $\alpha_1 + \dots + \alpha_d = 1$, the Schur-Weyl distribution $\mathrm{SW}^n(\alpha)$ is a distribution over partitions $\lambda \vdash n$ into at most $d$ parts where the probability of sampling a given $\lambda$ is $\dim(\lambda) \cdot s_{\lambda}(\alpha)$.  When $\alpha$ is uniform, we may write $\mathrm{SW}^n_d$ instead.
\end{definition}

\begin{fact}
For a quantum state $\rho \in \C^{d \times d}$ with eigenvalues $\alpha_1, \dots , \alpha_d$, if we perform weak Schur sampling on $\rho^{\otimes n}$, then the resulting distribution over $\lambda$ is exactly $\mathrm{SW}^n(\alpha)$.  
\end{fact}

One of the key tools for understanding the Schur-Weyl distribution is the following combinatorial characterization.

\begin{fact}[The RSK-Correspondence, Greene's Theorem \cite{romik2015surprising}]\label{fact:RSK}
Consider $\lambda \sim \mathrm{SW}^n(\alpha)$.  Alternatively, sample a sequence of $n$ tokens in $[d]$ say $x = (x_1, \dots , x_n)$ where each token is drawn independently from the distribution $(\alpha_1, \dots , \alpha_d)$.  For each $k \leq n$, let $l_k$ be the maximum length of the union of $k$ disjoint weakly increasing subsequences of $x$ and let $l'_k$ be the maximum length of the union of $k$ disjoint strictly decreasing subsequences of $k$.  Then we have that:   
\begin{itemize}
\item The distribution of $\lambda$ is the same as the distribution of $(l_1, l_2 - l_1, l_3 - l_2, \dots , l_{n} - l_{n-1})$
\item The distribution of $\lambda^T$ is the same as the distribution of $(l_1', l_2' - l_1', \dots, l_{n}' - l_{n-1}')$
\end{itemize}
\end{fact}

We now prove a few basic inequalities about the Schur-Weyl distribution.  Compared to \cite{o2016efficient,o2017efficient} where inequalities of a similar flavor are used, the inequalities here are stronger in the regime $n \ll d^2$, which will be crucial later on.

\begin{claim}\label{claim:typical-young-tableaux1}
Let $\alpha = (\alpha_1, \dots , \alpha_d)$ be a vector of nonnegative weights summing to $1$. Then for  $\lambda \sim \mathrm{SW}^n(\alpha)$, with probability at least $1/2$,
\[
\sum_{i = 1}^d \lambda_i^2 \geq \frac{n^{1.5}}{4} \,.
\]
\end{claim}
\begin{proof}
Consider sampling a sequence $x = (x_1, \dots , x_n)$ in $[d]^n$ where each token is drawn independently from the distribution $(\alpha_1, \dots , \alpha_d)$ as in Fact~\ref{fact:RSK}.  Now we pair up each sequence $x$ with its reverse, say $x'$.  These sequences occur with equal probability.  Next, let $\lambda(x)$ be the partition corresponding to $x$ as defined in Fact~\ref{fact:RSK}.  Note that by construction, $\lambda(x') \succ \lambda(x)^T$.  On the other hand, for any partition $\lambda \vdash n$, we claim that
\[
\max\left(\sum_{i = 1}^n \lambda_i^2 , \sum_{i = 1}^n (\lambda_i^T)^2 \right) \geq \frac{n^{1.5}}{4} \,. 
\]
Once we prove the above we are done, since $x$ and $x'$ occur with the same probability and then we would get that either $\lambda(x)$ or $\lambda(x')$ satisfy the desired property.  To see why the above holds, if $\lambda_1 + \dots + \lambda_{\sqrt{n}/2} \geq \frac{n}{2}$, then $\sum_{i = 1}^{\sqrt{n}/2} \lambda_i^2 \geq n^{1.5}/2$ by Cauchy Schwarz.  On the other hand, if $\lambda_1 + \dots + \lambda_{\sqrt{n}/2} \leq \frac{n}{2}$, then 
$\sum_{i > \sqrt{n}/2 } \lambda_i \geq \frac{n}{2}$ which actually implies that 
\[
\sum_{i, \lambda_i^T \geq \frac{\sqrt{n}}{2}} \lambda_i^T \geq \frac{n}{2} \,.
\]
The above then implies that  $\sum_{i = 1}^n (\lambda_i^T)^2 \geq \frac{n}{2} \cdot \frac{\sqrt{n}}{2} = \frac{n^{1.5}}{4}$.  This completes the proof.
\end{proof}

\begin{claim}\label{claim:typical-young-tableaux2}
Let $\alpha = (\alpha_1, \dots , \alpha_d)$ be a vector of nonnegative weights summing to $1$.  Then for  $\lambda \sim \mathrm{SW}^n(\alpha)$, 
\[
\E\left[ \sum_{i = 1}^d \lambda_i^2 \right] \leq 2((\alpha_1^2 + \dots + \alpha_d^2)n^2 + n^{1.5})
\]
\end{claim}
\begin{proof}
Define the polynomial
\[
p^*_2(\lambda) = \sum_{i = 1}^{\ell(\lambda)}\left( \left(\lambda_i - i + \frac{1}{2}\right)^2 - \left(i - \frac{1}{2}\right)^2 \right) \,.
\]
It is a standard fact (see \cite{o2015quantum}) that
\begin{equation}\label{eq:sw-moments}
\E_{\lambda \sim SW^n(\alpha)} \left[p^*_2(\lambda) \right] = n(n-1)\sum_{i = 1}^d \alpha_i^2 \,.
\end{equation}
Next, note that 
\[
\sum_{i = 1}^{\ell(\lambda)} (2i - 1)\lambda_i = \sum_{j = 1}^{\ell(\lambda_j^T)} (\lambda_j^T)^2 \,.
\]
Finally, recall from Fact~\ref{fact:RSK} that $\lambda_1^T$ has the same distribution as the longest strictly decreasing subsequence of a random sequence of $n$ tokens in $[d]$ drawn independently from the distribution $(\alpha_1, \dots , \alpha_d)$.  This is dominated by the longest decreasing subsequence of a random permutation and it is known (see e.g. \cite{wright2016learn}) that this has expected value at most $2\sqrt{n}$.  Thus,
\[
\E[\lambda_1^T ] \leq 2\sqrt{n}
\]
so
\[
\E\left[\sum_{i = 1}^{\ell(\lambda)} (2i - 1)\lambda_i \right] = \E\left[\sum_{j = 1}^{\ell(\lambda_j^T)} (\lambda_j^T)^2 \right] \leq 2n^{1.5} \,.
\]
Combining the above with \eqref{eq:sw-moments} implies
\[
\E\left[ \sum_{i = 1}^d \lambda_i^2 \right] \leq 2((\alpha_1^2 + \dots + \alpha_d^2)n^2 + n^{1.5})
\]
as desired.
\end{proof}

\subsection{Tensor manipulations}
We will also need some general notation for working with vectors, matrices and tensors.

\begin{definition}
For a vector $v \in \C^{d^t}$, we define $T(v)$ to be the $d \otimes  \dots \otimes d$ tensor obtained by reshaping $v$ (we assume that this is done in a canonical and consistent way throughout this paper).
\end{definition}

\begin{definition}
For matrices (or tensors) $A_1, \dots , A_t$ with all dimensions equal, we denote the symmetric sum 
\[
\sum_{\mathrm{sym}} A_1 \otimes \dots \otimes A_t = \sum_{\pi \in S_t} A_{\pi(1)} \otimes \dots \otimes A_{\pi(t)} \,.
\]
If we additionally have positive integers $k_1, \dots , k_t$, then 
\[
\sum_{\mathrm{sym}} A_1^{\otimes k_1} \otimes \dots \otimes A_t^{\otimes k_t} = \sum_{f \in \calS_{k_1, \dots , k_t}} A_{f(1)} \otimes \dots \otimes A_{f(k_1 + \dots + k_t)} \,.
\]
where $\calS_{k_1, \dots , k_t}$ consists of all distinct permutations of the multiset with $k_1$ elements equal to $1$ , $k_2$ elements equal to $2$ and so on. 
\end{definition}

\begin{definition}
Given a tensor $T \in (\C^d)^{\otimes t}$, we index its modes $1,2, \dots , t$.  For sets $S_1, \dots , S_k$ that partition $[t]$, we define $F_{S_1, \dots , S_k}(T)$ to be the order-$k$ tensor whose dimensions are $d^{|S_1|} \times \dots \times d^{|S_k|}$ and are obtained by flattening the respective modes of $T$ indexed by elements of $S_1, \dots , S_k$ respectively.    
\end{definition}

\begin{definition}\label{def:unfolded-matrix}
For a vector $v \in \C^{d^t}$ and integer $1 \leq j \leq t$, we define $G_j(v)$ to be a $d^j \times d^j$ matrix defined as 
\[
G_j(v) = \sum_{S \subset [t], |S| = j} \left(F_{S, [t]\backslash S}(T(v)) \right) \left(F_{S, [t]\backslash S}(T(v))\right)^\dagger \,.
\]
\end{definition}

Now we have the following relations.
\begin{fact}\label{fact:unfolded-inner-product}
For any vector $v \in \C^{d^t}$, matrix $E \in \C^{d \times d}$ and integer $1 \leq k \leq t$, 
\[
\langle \sum_{\mathrm{sym}} E^{\otimes k} \otimes I_d^{\otimes t - k} , vv^\dagger \rangle = \langle E^{\otimes k}, G_k(v)\rangle \,. 
\]
\end{fact}
\begin{proof}
The identity follows immediately from the definitions.   
\end{proof}

\begin{lemma}\label{lem:fake-estimator-mean}
 Let $f_1 \geq  \dots \geq  f_d \geq 0$ and  $f_1 + \dots + f_d = t$.  Let $v_1, \dots , v_d$ be an orthonormal basis for $\C^d$.  Let 
 \[
v = \sum_{g} w_{g} (v_{g(1)} \otimes \cdots \otimes v_{g(t)}) 
\]
where $g: [t] \rightarrow [d]$ ranges over all functions such that exactly $f_j$ distinct elements are mapped to $j$ for each $j \in [d]$.  Then 
\[
G_1(v) = \left(\sum_g w_g^2 \right) \cdot  (f_1 v_1v_1^\dagger + \dots + f_d v_dv_d^\dagger) \,.
\]
\end{lemma}
\begin{proof}
Note that for any two distinct functions $g, g'$ with the specified frequencies, for any index $j$,
\[
F_{\{j \}, [t] \backslash \{j \}}( v_{g(1)} \otimes \dots \otimes v_{g(t)})  F_{\{j \}, [t] \backslash \{j \}}( v_{g'(1)} \otimes \dots \otimes v_{g'(t)})^\dagger  = v_{g(j)}v_{g'(j)}^\dagger \cdot \prod_{l \neq j } \langle v_{g(l)},  v_{g'(l)} \rangle =  0 
\]
because $g, g'$ must differ on at least two inputs $l \in [t]$ (since their outputs have the same frequencies) so then $\langle v_{g(l)}, v_{g'(l)}\rangle = 0$ for some $l \in [t]\backslash \{j \}$.  Next, using the same formula as above but with $g = g'$,
\[
\begin{split}
F_{\{j \}, [t] \backslash \{j \}}( v_{g(1)} \otimes \dots \otimes v_{g(t)}) F_{\{j \}, [t] \backslash \{j \}}( v_{g(1)} \otimes \dots \otimes v_{g(t)})^\dagger = v_{g(j)}v_{g(j)}^\dagger\,.
\end{split}
\]
Writing out the definition of $v$ and combining the above two identities completes the proof of the desired equality.
\end{proof}

As a consequence of Lemma~\ref{lem:fake-estimator-mean}, we have:
\begin{corollary}\label{coro:unfolded-keyl-povm}
Let $d$ be an integer and $\lambda \vdash t$ be a partition.  Let $v_{\lambda,j}^d \in \C^{d^t}$ be one of the vectors as defined in Definition~\ref{def:weight-vector}.  Then
\[
G_1(v_{\lambda,j}^d) = \diag(\lambda_1, \dots , \lambda_d) \,.
\]
Furthermore, for any unitary $U \in \C^{d \times d}$, 
\[
G_1(U^{\otimes t}v_{\lambda,j}^d) = U\diag(\lambda_1, \dots , \lambda_d)U^\dagger \,.
\]
\end{corollary}
\begin{proof}
This statement follows immediately from Lemma~\ref{lem:weight-vector} and Lemma~\ref{lem:fake-estimator-mean}.
\end{proof}

\subsection{Additional facts}

In this section, we present a few additional facts that will be used later in the proof. 
 First, we give an expression for the symmetric polynomials as a linear combination of products of power sums.

\begin{lemma}\label{lem:elem-power-sums}
Let $x = (x_1, \dots , x_d)$  be a $d$-tuple of variables and let $\lambda \vdash n$ be a partition of $n$ into $d$ parts (possibly including zeros).  For each integer $k$, let $S_k = x_1^k + \dots + x_d^k$.  Then we can write
\[
\mathsf{count}(\lambda) \sum_{\mathrm{sym}} x^{\lambda} = \sum_{\substack{\mu \vdash n \\ \mu \succeq \lambda} }c_{\mu} \prod_{\mu_i > 0} S_{\mu_i}
\]
for some coefficients $c_{\mu}$ satisfying $| c_{\mu}| \leq (2n^2)^{(\ell(\lambda) - \ell(\mu))}$ (recall $\mathsf{count}$ is defined in Definition~\ref{def:count}).
\end{lemma}
\begin{proof}
We prove the statement by induction on $\lambda$ where we induct on the reverse ``majorizing order".  The base case is when $\lambda = (n,0, \dots , 0)$, which is the (unique) maximal partition according to the majorizing order.  Now assume that we have proved the claim for all $\lambda'$ that strictly majorize $ \lambda$.  Now note that we can write
\[
\prod_{\lambda_i > 0} S_{\lambda_i} = \sum_{\substack{\mu \vdash n \\ \mu \succeq \lambda} } z_{\mu} \sum_{\mathrm{sym}} x^{\mu}
\]
for some coefficients $z_{\mu}$ that are nonnegative integers.  Furthermore,  the coefficient $z_{\lambda}$ is exactly equal to $\mathsf{count}(\lambda)$.  Thus, 
\[
\mathsf{count}(\lambda) \sum_{\mathrm{sym}} x^{\lambda} - \prod_{\lambda_i > 0} S_{\lambda_i}
\]
is a symmetric polynomial whose monomials strictly majorize $\lambda$ and we can now apply the inductive hypothesis to write $\mathsf{count}(\lambda) \sum_{\mathrm{sym}} x^{\lambda}$ in the desired form.  Using the inductive hypothesis, we can bound the coefficients.  First $c_{\lambda} = 1$ by construction.  Next, for $\mu \succ \lambda$, we can bound $c_{\mu}$ as 
\[
|c_{\mu}| \leq  \sum_{\substack{\mu' \vdash n \\ \mu \succeq \mu' \succ \lambda}} \frac{(2n^2)^{\ell(\mu') - \ell(\mu)} z_{\mu'}}{\mathsf{count}(\mu')} \leq (2n^2)^{\ell(\lambda) - \ell(\mu)} \sum_{\substack{\mu' \vdash n \\  \mu' \succ \lambda} } \frac{z_{\mu'}}{\mathsf{count}(\mu') (2n^2)^{\ell(\lambda) - \ell(\mu')}} \,.
\]
Let $\ell(\lambda) = k$ and fix some $\mu'$ with $\ell(\mu') = k'$. Now recall that $z_{\mu'}$ is the coefficient of $x^{\mu'}$ in the expansion $S_{\lambda_1} \cdots S_{\lambda_k}$.  It is $0$ unless $k' < k$ (since we are assuming $\mu' \succ \lambda$).  Also, $z_{\mu'}/\mathsf{count}(\mu')$ is at most equal to the number of partitions of the set $\{ \lambda_1, \dots , \lambda_k \}$ into $k'$ parts such that the sums of the parts of the partition are exactly $\mu'_1, \dots , \mu'_{k'}$.  Thus,  $\sum_{\ell(\mu') = k'} z_{\mu'}/\mathsf{count}(\mu')$ is at most equal to the number of ways to partition $[k]$ into exactly $k'$ disjoint subsets.  This is at most equal to the number of labeled forests on $k$ vertices with $k - k'$ edges which is at most $k^{2(k - k')} \leq (n^2)^{\ell(\lambda) - \ell(\mu')}$.  Substituting this into the above sum completes the proof.  
\end{proof}

We also have the following formulas for integrals over Haar random unitaries.

\begin{claim}\label{claim:haar-integrals}
Let $v \in \C^d$ be a uniformly random unit vector.  Let $v_1$ be its first entry and $v_2$ be its second entry.  Then
\[
\begin{split}
&\E[|v_1|^2] = \E[|v_2|^2] = \frac{1}{d} \\
&\E[|v_1|^4] = \frac{2}{d(d+1)}  \\
&\E[|v_1|^2 |v_2|^2] = \frac{1}{d(d+1)} \,.
\end{split}
\]
\end{claim}
\begin{proof}
These follow from direct computation using formulas for Haar integrals (see \cite{harrow2013church}).
\end{proof}

\begin{claim}\label{claim:haar-moments}
Let $X,Y \in \C^{d \times d}$ be Hermitian matrices.  Then
\[
\E_{U}[(U^\dagger X U) \langle U^{\dagger}X U, Y \rangle  ] =  \frac{1}{d^2 - 1}\left(\norm{X}_F^2 - \frac{\tr(X)^2}{d} \right) \left( Y - \frac{\tr(Y) I}{d}\right) + \frac{\tr(X)^2 \tr(Y) I}{d^2}
\]
where the expectation is over a Haar random unitary $U$.
\end{claim}
\begin{proof}
First by Claim~\ref{claim:haar-integrals}, for a random unit vector $v$,
\[
\E_{v}[vv^\dagger \langle vv^\dagger , Y \rangle  ] = \frac{\tr(Y)I + Y}{d(d+1)} 
\]
and for Haar random orthonormal unit vectors $v_1, v_2$,
\[
\E_{\substack{v_1,v_2 \\ v_1 \perp v_2}}[v_1v_1^\dagger \langle v_2v_2^\dagger , Y \rangle  ] = \E_{\substack{v_1,v_2 \\ v_1 \perp v_2}}\left[v_1v_1^\dagger \left\langle \frac{I - v_1v_1^\dagger}{d-1} , Y \right\rangle  \right] =  \frac{d\tr(Y) I - Y}{(d-1)d(d+1)} \,.
\]
Now let the eigenvalues of $X$ be $x_1, \dots , x_d$ and the eigenvalues of $Y$ be $y_1, \dots y_d$.  We have
\[
\begin{split}
\E_{U}[(U^\dagger X U) \langle U^{\dagger}X U, Y \rangle  ] &=  \E_{v} \left[ \left(\sum_j x_j^2\right) vv^\dagger \langle vv^\dagger , Y \rangle \right]  + \E_{\substack{v_1,v_2 \\ v_1 \perp v_2}}\left[ \left( 2\sum_{j_1 < j_2} x_{j_1}x_{j_2}\right)v_1v_1^\dagger \langle v_2v_2^\dagger , Y \rangle  \right] \\ &= \norm{X}_F^2 \cdot \frac{\tr(Y)I + Y}{d(d+1)}  + (\tr(X)^2 - \norm{X}_F^2 ) \frac{d\tr(Y) I - Y}{(d-1)d(d+1)} \\ &= \frac{1}{d^2 - 1}\left(\norm{X}_F^2 - \frac{\tr(X)^2}{d} \right) \left( Y - \frac{\tr(Y) I}{d}\right) + \frac{\tr(X)^2 \tr(Y) I}{d^2} \,.
\end{split}
\]
\end{proof}

\section{Learning balanced states}\label{sec:learn-balanced}

Now, we will present our learning algorithm for states that are close to maximally mixed.  The main result that we prove in this section is as follows:
\begin{theorem}\label{thm:learn-balanced}
Let $\rho = \frac{I_d }{d} + E$ be an unknown quantum state in $\C^{d \times d}$.  Let $t$ be a parameter such that $t \leq d^2$ and $\norm{E}_F \leq (0.01/t)^4$.  Then for any target accuracy $\eps$, there is an algorithm (Algorithm~\ref{alg:learn-balanced}) that takes $O(d^2/(t^{1.5}\eps^2))$ copies of $\rho^{\otimes t}$ and  returns $\wh{E} \in \C^{d \times d}$ such that 
\[
\E\left[\norm{E - \wh{E}}_F^2 \right] \leq O(t^3 \norm{E}_F^4 + \eps^2 ) \,.
\]
\end{theorem}

Our algorithm will make use of Keyl's POVM (recall Definition~\ref{def:keyl-povm}).

\begin{algorithm2e}[ht!]\label{alg:learn-balanced}
\caption{Algorithm for Learning Balanced States}
\DontPrintSemicolon
\textbf{Input:} $m$ copies of $\rho^{\otimes t}$ for some unknown quantum state $\rho \in \C^{d \times d}$
\;
\For{$j \in [m]$}{

Measure $\rho^{\otimes t}$ according to Keyl's POVM  \;
Let $\lambda \vdash t$ be the partition and $U$ be the unitary obtained from the measurement\;
Set $D_j = U \diag(\lambda_1/t, \dots , \lambda_d/t)U^{\dagger}$ \;
}
Compute $\theta = \E_{\lambda \sim \mathrm{SW}^t_d}[\sum_j \lambda_j^2 ] - (t^2/d)  $ \;\label{step:theta}
Compute $\wh{E} =  \frac{t(d^2 - 1)}{d\theta}\left(\frac{D_1 + \dots + D_m}{m} - \frac{I_d}{d}\right)$\;
\textbf{Output:} $\wh{E}$\;

\end{algorithm2e}

\noindent We can actually use Theorem~\ref{thm:learn-balanced} to learn any state $\rho$ with all eigenvalues at most $4/d$ \footnote{the choice of the constant $4$ is arbitrary} by first constructing an estimate $\sigma$ for the ``complement" $\frac{4}{3d} I_d - \frac{1}{3}\rho $ and then simulating measurement access to copies $(\rho + 3\sigma)/4$ and applying Theorem~\ref{thm:learn-balanced} to refine the error in the estimate $\sigma$.  To initially estimate $\sigma$, we will use the standard algorithm for state tomography with unentangled measurements. 
 
 \begin{theorem}[\cite{guctua2020fast}]\label{thm:untentangled-tomography}
For any $\delta, \eps < 1$ and unknown state $\rho \in \C^{d \times d}$, there is an algorithm that makes unentangled measurements on $O(d^2 \log(1/\delta)/\eps^2  )$ copies of $\rho$ and with $1 - \delta$ probability outputs a state $\wh{\rho}$ such that $\norm{\rho - \wh{\rho}}_F \leq \eps$.
\end{theorem}

\noindent In this way, we can prove the following corollary of Theorem~\ref{thm:learn-balanced}.  It is actually slightly more general in that we can learn any state $\rho$ that can be written as a sum $\sigma + \Delta$ where $\norm{\sigma} \leq 4/d$ and $\norm{\Delta}_F$ is small.

\begin{corollary}\label{coro:learn-balanced}
Let $\delta, \eps < 1$ and $\rho \in \C^{d \times d}$ be an unknown state.  Let $t \leq \min(d^2, c(1/\eps)^{0.2} )$ be some parameter, where $c$ is a sufficiently small absolute constant.  Assume that $\rho$ can be written as $\rho = \rho' + \Delta$ where $\norm{\rho'} \leq 4/d$ and $\norm{\Delta}_F \leq \sqrt{\eps}/t^2$.  Then there is an algorithm that takes $O(d^2 \log(1/\delta)/(\sqrt{t}\eps^2))$ total copies of $\rho$ and measures them in batches of $t$-entangled copies and with probability $1 - \delta$, outputs a state $\wh{\rho}$ such that $\norm{\rho - \wh{\rho}}_F \leq \eps$.
\end{corollary}
\begin{proof}
    In the first step, we run the algorithm in Theorem~\ref{thm:untentangled-tomography} with $O(d^2\log(1/\delta)/(\sqrt{t}\eps^2))$ samples to produce an estimate $\hat{\rho}$ satisfying $\norm{\hat{\rho} - \rho}_F \le O(\eps t^{1/4})$. Project this (with respect to the Frobenius norm) to the convex set of density matrices with operator norm at most $4/d$ to produce a new estimate $\wt{\rho}$. Note that $\norm{\wt{\rho} - \rho'}_F \le \norm{\hat{\rho} - \rho'}_F$, and
    \begin{equation*}
        \norm{\wt{\rho} - \rho}_F \le \norm{\wt{\rho} - \rho'}_F + \norm{\Delta}_F \le \norm{\wh{\rho} - \rho'}_F + \norm{\Delta}_F \le \norm{\wh{\rho} - \rho}_F + 2\norm{\Delta}_F \le O(\eps t^{1/4} + \sqrt{\eps}/t^2) \le O(\sqrt{\eps}/t^2)\,,
    \end{equation*}
    where in the last step we used the assumption that $t \le c(1/\eps)^{0.2}$. Because $\norm{\wt{\rho}} \le 4/d$ by design, the ``complement'' state $\sigma \triangleq \frac{4}{3d}I_d - \frac{1}{3}\wt{\rho}$ is a valid density matrix. We would like to apply Theorem~\ref{thm:learn-balanced} to the state $(\rho + 3\sigma)/4$. Note that
    \begin{equation}
        \frac{1}{4}(\rho+ 3\sigma) = \frac{1}{d}I_d + \frac{1}{4}(\rho - \wt{\rho})\,,
    \end{equation}
    and for $E\triangleq (\rho - \wt{\rho})/4$, we have $\norm{E}_F \le O(\sqrt{\eps}/t^2) \le (0.01/t)^4$ by the assumption that $t\le c(1/\eps)^{0.2}$, provided we take $c$ sufficiently small. Theorem~\ref{thm:learn-balanced} thus implies that with $O(d^2/(t^{1.5}\eps^2))$ copies of $((\rho + 3\sigma)/4)^{\otimes t}$, measured in batches of $t$-entangled copies, we can use Algorithm~\ref{alg:learn-balanced} to produce an estimate $\wh{E}$ for which 
    \begin{equation}
        \norm{E - \wh{E}}_F \le O(t^{3/2}\norm{E}^2_F + \eps) \le O(\eps)    
    \end{equation}
    with probability at least $1 - \delta$, where in the last step we used that $\norm{E}^2_F \le O(\eps / t^4)$. By Lemma~\ref{lem:sim}, we can simulate $t$-entangled measurement access to $(\rho + 3\sigma)/4$ using just $t$-entangled measurement access to $\rho$.

    Our estimate for $\rho$ is given by taking $\wt{\rho} + 4\wh{E}$ and projecting to the convex set of density matrices. As this projection can only decrease the Frobenius distance to $\rho$, the proceeding arguments imply that 
    \[
    \norm{\wt{\rho} + 4\wh{E} - \rho}_F \leq O(\eps)
    \]
    with constant probability with using $O(d^2 / (\sqrt{t} \eps^2))$ copies.  To achieve failure probability $1 - \delta$, we can simply repeat this process $O(\log 1 / \delta)$ times and take the geometric median of the outputs.
    Standard arguments then imply that this is $O(\eps)$ close to $\rho$ with probability $1 - \delta$.  Rescaling  $\eps$ by an appropriate constant factor completes the proof. 
    \end{proof}

\begin{lemma}\label{lem:sim}
    Let $0 \le \lambda \le 1$. Given $t$ copies of an unknown state $\rho$, and given a description of a density matrix $\sigma$, it is possible to simulate any measurement of $(\lambda\rho + (1 - \lambda)\sigma)^{\otimes t}$ using a measurement of $\rho^{\otimes t}$.
\end{lemma}

\begin{proof}
    Let $\{M_z\}_{z\in\calZ}$ be an arbitrary POVM. Then for any $z\in\calZ$, observe that
    \begin{equation}
        \langle M_z, (\lambda \rho + (1 - \lambda)\sigma)^{\otimes t}\rangle = \sum_{S\subseteq[t]} \lambda^{|S|}(1 - \lambda)^{t-|S|}\cdot \langle M_z, \rho\otimes_S \sigma\rangle\,,
    \end{equation}
    where $\rho\otimes_S \sigma$ denotes the product state such that for every $i\in S$, the $i$-th component of the state is a copy of $\rho$, and every other component is a copy of $\sigma$. To simulate measuring $(\lambda\rho + (1-\lambda)\sigma)$ with $\{M_z\}_{z\in\calZ}$, we can simply sample $S$ by including each $i\in[t]$ in $S$ with probability $\lambda$, prepare the state $\rho\otimes_S \sigma$, and measure with $\{M_z\}$.
\end{proof}

\noindent The remainder of this section will be devoted to proving Theorem~\ref{thm:learn-balanced}.

\subsection{Approximation for rotationally invariant POVMs}

Recall that our goal in this section is to learn $\rho$ when $\rho = I_d/d + E$ for some $E$ that is sufficiently small.  One important insight in the analysis of Algorithm~\ref{alg:learn-balanced} is that because Keyl's POVM is invariant to rotating all of the copies simultaneously, we can replace $X = \rho^{\otimes t} = (I_d/d + E)^{\otimes t}$ with its first order approximation $X' = (I_d/d)^{\otimes t} + \sum_{\mathrm{sym}} E \otimes (I_d/d)^{\otimes t - 1}$ at the cost of some small error.  We can then analyze the POVM applied to $X'$ instead which is significantly simpler (although $X'$ may not technically be a state). 

First, we have the following lemma for bounding quantities that involve averaging over simultaneous Haar random rotations of all copies.
\begin{lemma}\label{lem:random-rotations}
Let $v \in \C^{d^t}$ be any vector with $\norm{v} = 1$.  Let $E \in \C^{d \times d}$ be a Hermitian matrix with $\tr(E) = 0$.  Then 
\[
\left\lvert \E_{U}[\langle (U E U^{\dagger})^{\otimes t}, vv^\dagger \rangle  ] \right\rvert \leq \frac{(4t)^{4t} \norm{E}_F^t}{d^t}
\]
where $U$ is a Haar random unitary.
\end{lemma}
\begin{proof}
Let the eigenvalues of $E$ be $\alpha_1, \dots , \alpha_d$.  Then
\[
\E_{U}[\langle (U^\dagger E U)^{\otimes t}, vv^\dagger \rangle ]
\]
is a symmetric, degree $t$ polynomial in the $\alpha_i$ since $U$ is Haar random.  Now consider any tuple of nonnegative integers $\lambda = (\lambda_1, \dots , \lambda_d)$ with $\lambda_1 \geq \dots \geq \lambda_d$ and $\lambda_1 + \dots + \lambda_d = t$.  For each integer $k$, define $S_k = \alpha_1^k + \dots + \alpha_d^k$.  Use Lemma~\ref{lem:elem-power-sums} to write
\[
\mathsf{count}(\lambda) \sum_{\mathrm{sym}} \alpha^{\lambda} = \sum_{\substack{\mu \vdash t \\ \mu \succeq \lambda }} c_{\mu} \prod_{\mu_i > 0} S_{\mu_i} \,.
\]
Next note that $S_1 = 0$ and for all $k \geq 2$, $S_k \leq (\alpha_1^2 + \dots + \alpha_d^2)^{k/2}$. 
Thus, 
\[
\left\lvert \mathsf{count}(\lambda) \sum_{\mathrm{sym}} \alpha^{\lambda} \right \rvert \leq (\alpha_1^2 + \dots + \alpha_d^2)^{t/2}  (10t^2)^{t} 
\]
where we used the coefficient bound in Lemma~\ref{lem:elem-power-sums} and Fact~\ref{fact:num-partitions}.  Next, recall that by Schur Weyl duality, we can write
\[
\E_{U} [ (U^\dagger E U)^{\otimes t}] = \sum_{\lambda \vdash t} \frac{s_{\lambda}( \alpha_1, \dots , \alpha_d)}{\dim(V_{\lambda}^t)} \Pi_{\lambda}^t \,.
\]
All of the Schur subspaces $V_{\lambda}^t$ have dimension at least $(d/t)^t$ (to see this, note that if $t < d$, then there are at least $\binom{d}{t}$ distinct SSYT of any given shape and otherwise the claim is vacuously true).  Thus, since $\norm{v} = 1$, we can write
\begin{equation}\label{eq:sumton}
\left\lvert \E_{U}[\langle (U^\dagger E U)^{\otimes t}, vv^\dagger \rangle ] \right\rvert  = \sum_{\lambda} z_{\lambda} s_{\lambda}(\alpha_1, \dots , \alpha_d)
\end{equation}
where the coefficients $z_{\lambda}$ satisfy $\sum_{\lambda} |z_{\lambda}| \leq (t/d)^t$. 
Finally, note that the Schur polynomial $s_{\lambda}(\alpha_1, \dots , \alpha_d)$ has positive coefficients and these are all dominated, monomial by monomial, by the coefficients in $(\alpha_1 + \dots + \alpha_d)^t$ which are all at most $t^t$.  Thus, overall we conclude by Eq.~\eqref{eq:sumton} that
\begin{align*}
\left\lvert \E_{U}[\langle (U^\dagger E U)^{\otimes t}, vv^\dagger \rangle ] \right\rvert &\leq |\{ \lambda: \lambda \vdash t \} | \cdot  t^t \cdot (t/d)^t  \cdot (10t^2)^{t}  (\alpha_1^2 + \dots + \alpha_d^2)^{t/2}  \leq \frac{(4t)^{4t} \norm{E}_F^t}{d^t} \,. \qedhere
\end{align*}
\end{proof}

We now define the family of POVMs in $\C^{d^t \times d^t}$ that are invariant to rotating all of the $d \times d$ copies simultaneously by the same unitary. 

\begin{definition}
We say a POVM $\{ M_z \}_{z \in \calZ}$ in $\C^{d^t \times d^t}$ is copy-wise rotationally invariant if it is equivalent to 
\[
\{ U^{\otimes t} M_z (U^{\dagger})^{\otimes t} dU\}_{ z \in \calZ}
\]
where $U \in \C^{d \times d}$ is a random unitary drawn from the Haar measure. 
\end{definition}

\noindent Next, we bound a $\chi^2$-like distance between the outcome distributions from measuring $X$ and $X'$ with a copy-wise rotationally invariant POVM. Again, we emphasize that the low-degree truncation $X'$ is not necessarily an actual density matrix; nevertheless, the expression on the left-hand side of Eq.~\eqref{eq:chisq} below is a well-defined quantity.

\begin{lemma}\label{lem:replace-with-first-order}
Let $\{ M_z \}_{z \in \calZ}$ be a POVM in $\C^{d^t \times d^t}$ that is copywise rotationally invariant.  Let $E \in \C^{d \times d}$ be a matrix with $\tr(E) = 0$ and $\norm{E}_F \leq \left(\frac{0.01}{t}\right)^4$.  Let $X = (I_d/d + E)^{\otimes t}$ and let $X' = (I_d/d)^{\otimes t} + \sum_{\mathrm{sym}} E \otimes (I_d/d)^{\otimes t-1}$.  Then
\begin{equation}
\int_{\mathcal{Z}} \frac{ d^t\langle X - X', M_z \rangle^2 }{\tr(M_z) } \,\mathrm{d}z \leq (100t)^4 \norm{E}_F^4 \,. \label{eq:chisq}
\end{equation}
\end{lemma}

\begin{proof}
We write 
\[
X =  X' + \sum_{j = 2}^t \sum_{\mathrm{sym}} E^{\otimes j} \otimes (I_d/d)^{\otimes t - j}  \,.
\]
Now we will first bound 
\[
\begin{split}
&\int_{\calZ} \frac{d^t\langle M_z , E^{\otimes j} \otimes (I_d/d)^{\otimes t - j} \rangle^2 }{\tr(M_z)} \,\mathrm{d}z = \int_{\calZ} \frac{ d^t\langle M_z \otimes M_z , E^{\otimes j} \otimes (I_d/d)^{\otimes t - j} \otimes E^{\otimes j} \otimes (I_d/d)^{\otimes t - j}\rangle }{\tr(M_z)}\,\mathrm{d}z \\ &= \int_{\calZ} \E_U \left[ \frac{d^t \langle M_z \otimes M_z , U^{\otimes 2t}(E^{\otimes j} \otimes (I_d/d)^{\otimes t - j} \otimes E^{\otimes j} \otimes (I_d/d)^{\otimes t - j})(U^{\dagger})^{\otimes 2t}\rangle}{\tr(M_z)} \right] \,\mathrm{d}z\,,
\end{split}
\]
where the expectation is over Haar random unitaries $U \in \C^{d \times d}$ (the last step is valid because we assumed that the POVM is copywise rotationally invariant).  Now we can use Lemma~\ref{lem:random-rotations} to conclude that the above is at most 
\[
\int_{\calZ} d^t \tr(M_z) \cdot  \frac{(8j)^{8j} \norm{E}_F^{2j}}{d^{2t}}\,\mathrm{d}z = (8j)^{8j} \norm{E}_F^{2j} \,.
\]
Thus, by Cauchy~Schwarz we have
\begin{align*}
\MoveEqLeft\int_{\calZ} \frac{ d^t\langle X - X', M_z \rangle^2 }{\tr(M_z) }\,\mathrm{d}z \\
&\leq \int_{\calZ} \frac{d^t}{\tr(M_z)} \left(\sum_{S \subset [t], |S| \geq 2} \frac{1}{2^{|S|} \binom{t}{|S|}} \right)\left(\sum_{S \subset [t], |S| \geq 2} 2^{|S|}\binom{t}{|S|}  \langle M_z , E^{\otimes S} \otimes (I_d/d)^{\otimes [t] \backslash S} \rangle^2   \right) \,\mathrm{d}z\\ &\leq \sum_{j = 2}^t 2^{j} \binom{t}{j}^2 (8j)^{8j} \norm{E}_F^{2j} \leq (100t)^4 \norm{E}_F^4 \,,
\end{align*}
where we use the condition that $\norm{E}_F \leq (0.01/t)^4$ in the last step.
\end{proof}

\noindent Now given a copy-wise rotationally invariant POVM, we can consider a family of estimators $f: \C^{d^t \times d^t} \rightarrow \C^{d \times d} $ that are rotationally compatible with it.
\begin{definition}
Let $ \{ M_z \}_{z \in \calZ}$ be a POVM in $\C^{d^t \times d^t}$   that is copywise rotationally invariant.  We say a function $f:\{ M_z \}_{z \in \calZ} \rightarrow \C^{d \times d}$ is rotationally compatible with the POVM if  
\[
f(U^{\otimes t} M_z (U^{\dagger})^{\otimes t}) = U f(M_z) U^{\dagger}  
\]   
for all $z \in \calZ$ and unitary $U$.
\end{definition}

\noindent We can apply Lemma~\ref{lem:replace-with-first-order} to bound the error of a rotationally compatible estimator between measuring $X'$ and $X$ with some copy-wise rotationally invariant POVM.  Note that in Algorithm~\ref{alg:learn-balanced}, the POVM is rotationally invariant and our estimator is rotationally compatible with it.

\begin{lemma}\label{lem:estimator-error}
Let $ \{ M_z \}_{z \in \calZ}$ be a POVM in in $\C^{d^t \times d^t}$   that is copywise rotationally invariant.  Let $f:\{ M_z \}_{z \in \calZ} \rightarrow \C^{d \times d}$ be a rotationally compatible estimator such that $\tr(f(M_z)) = 0$ for all $z\in\mathcal{Z}$.  Let $X = (I_d/d + E)^{\otimes t}$ and let $X' = (I_d/d)^{\otimes t} + \sum_{\mathrm{sym}} E \otimes (I_d/d)^{\otimes t-1}$.  Assume that $\norm{E}_F \leq \left(\frac{0.01}{t} \right)^4$.  Then
\[
\norm{\int_{\calZ} f(M_z) \langle X - X', M_z \rangle \,\mathrm{d}z}_{F} \leq \frac{10^5 t^2 \norm{E}_F^2}{d} \sqrt{\int_{\calZ} \frac{\norm{f(M_z)}_F^2 \tr(M_z)}{d^t} \,\mathrm{d}z} \,.
\]
\end{lemma}
\begin{proof}
We will upper bound the left-hand side by bounding $\langle A, \int_{\mathcal{Z}} f(M_z)\langle X - X', M_z\rangle \,\mathrm{d}z\rangle$ for all $\norm{A}_F \le 1$. Note that the integral on the left-hand side has trace 0, so it suffices to consider $A$ with trace $0$.

Consider a matrix $A \in \C^{d \times d}$ with $\tr(A) = 0$.  Then
\[
\int_{\calZ} \langle A, f(M_z) \rangle \langle X - X', M_z \rangle \,\mathrm{d}z \leq \left( \int_{\calZ}  \langle A, f(M_z) \rangle^2 \frac{\tr(M_z)}{d^t}\,\mathrm{d}z \right)^{1/2} \left( \int_{\calZ} \frac{ d^t \langle X- X', M_z \rangle^2}{\tr(M_z)}\,\mathrm{d}z\right)^{1/2} \,.
\]
To bound the first term above, by Claim~\ref{claim:haar-moments}, for any traceless Hermitian matrix $Y \in \C^{d \times d}$,
\[
\E_U[ \langle A , U^\dagger Y U \rangle^2 ] = \frac{\norm{A}_F^2 \norm{Y}_F^2}{d^2 - 1} 
\]
where the expectation is over the Haar measure.  Thus, since the POVM $\{M_z \}_{z \in \calZ}$ is copywise rotationally invariant, we have
\[
\int_{\calZ}  \langle A, f(M_z) \rangle^2 \frac{\tr(M_z)}{d^t} \,\mathrm{d}z \leq \int_{\calZ} \frac{2\norm{A}_F^2 \norm{f(M_z)}_F^2}{d^2} \cdot \frac{\tr(M_z)}{d^t} \,\mathrm{d}z  \,.
\]
Next, by Lemma~\ref{lem:replace-with-first-order}
\[
\int_{\calZ} \frac{ d^t\langle X - X', M_z \rangle^2 }{\tr(M_z) } \,\mathrm{d}z \leq (100t)^4 \norm{E}_F^4
\]
and thus, taking the maximum over all $A$ with $\norm{A}_{F} \leq 1$, we get
\begin{align*}
\norm{\int_{\calZ} f(M_z) \langle X- X', M_z \rangle \,\mathrm{d}z}_{F} &\leq \frac{10^5 t^2 \norm{E}_F^2}{d} \sqrt{\int_{\calZ} \frac{\norm{f(M_z)}_F^2 \tr(M_z)}{d^t}\,\mathrm{d}z}\,.\qedhere 
\end{align*}
\end{proof}

\subsection{Proof of Theorem~\ref{thm:learn-balanced}}

Recall that the high-level idea for proving Theorem~\ref{thm:learn-balanced} is to replace measurements of $X = \rho^{\otimes t}$ with measurements of $X'$ and then analyze those measurements of $X'$.  The next result allows us to compute the mean of the estimator in Algorithm~\ref{alg:learn-balanced} if we were able to measure $X'$.

\begin{corollary}\label{coro:fake-estimator-mean}
Let $\{M_{\lambda, U} \}_{\lambda, U}$ be Keyl's POVM where $\lambda$ ranges over partitions of $t$ and $U$ ranges over unitaries in $\C^{d \times d}$.  Let $X' = (I_d/d)^{\otimes t} + \sum_{\mathrm{sym}} E \otimes (I_d/d)^{\otimes t-1}$.  Then
\[
\sum_{\lambda\vdash t}\int U \diag(\lambda_1/t, \dots , \lambda_d/t)U^\dagger \cdot \langle M_{\lambda ,U}, X'  \rangle \,\mathrm{d}U = \frac{I_d}{d} + \frac{dE}{t(d^2 - 1)}  \E_{\lambda \sim \mathrm{SW}^t_d}\Bigl[\sum_{j = 1}^d \lambda_j^2 - (t^2/d) \Bigr] \,.
\]
\end{corollary}
\begin{proof}
Note that the actual POVM elements of Keyl's POVM are $\dim(V_{\lambda}^t) U^{\otimes t} M_{\lambda}  (U^\dagger)^{\otimes t}$ where $\lambda$ ranges over all partitions and $U$ ranges over Haar random unitaries. Now we apply Fact~\ref{fact:unfolded-inner-product} and Corollary~\ref{coro:unfolded-keyl-povm} to the vectors $U^{\otimes t} v_{\lambda,j}$ as $j$ ranges over all of the components of $M_{\lambda}$ (recall Definition~\ref{def:weight-vector}).  Then
\[
\begin{split}
\langle M_{\lambda,U}, X'\rangle &= \sum_{j = 1}^{\dim(\lambda)} \langle  \dim(V_{\lambda}^t) U^{\otimes t} v_{\lambda, j} v_{\lambda, j}^\dagger (U^\dagger)^{\otimes t} , X' \rangle \\ &= \dim(\lambda) \dim(V_{\lambda}^t)\Bigl(\frac{1}{d^t} + \frac{1}{d^{t-1}}\langle U \diag(\lambda_1, \dots , \lambda_d)U^\dagger, E \rangle \Bigr) \,.
\end{split}
\]
So by taking $X = \diag(\lambda_1/t,\ldots,\lambda_d/t)$ and $Y = E$ in Claim~\ref{claim:haar-moments}, for any fixed $\lambda$ we get
\[
\E_{U}[U \diag(\lambda_1/t, \dots , \lambda_d/t)U^\dagger \cdot \langle M_{\lambda ,U}, X'  \rangle  ] = \dim(\lambda)\dim(V_{\lambda}^t)\Bigl( \frac{ I_d}{d^{t+1}} +  \frac{\sum_j \lambda_j^2 - (t^2/d)}{t(d^2 - 1)d^{t-1}} E \Bigr) \,.
\]
Now summing over all $\lambda$, we conclude
\begin{align*}
\sum_{\lambda\vdash t}\int \diag(\lambda_1/t, \dots , \lambda_d/t)U^\dagger \cdot \langle M_{\lambda ,U}, X'  \rangle \,\mathrm{d}U &= \frac{I_d}{d} + \frac{dE}{t(d^2 - 1)}  \E_{\lambda \sim \mathrm{SW}^t_d}\Bigl[\sum_j \lambda_j^2 - (t^2/d)\Bigr]  \,. \qedhere
\end{align*}
\end{proof}

\noindent Now we can complete the proof of Theorem~\ref{thm:learn-balanced}.
\begin{proof}[Proof of Theorem~\ref{thm:learn-balanced}]
Note that the POVM in Algorithm~\ref{alg:learn-balanced} is clearly copywise rotationally invariant and the estimator is rotationally compatible with it.  Let us use the shorthand $\{ M_z \}_{z \in \calZ}$ to denote this POVM and for $M_z$ corresponding to unitary $U$ and partition $\lambda$, we let $f(M_z) = U\diag(\lambda_1/t, \dots , \lambda_d/t)U^\dagger$.  We have
\[
\E\left[ \frac{D_1 + \dots + D_m}{m} \right] = \int_{\calZ} f(M_z) \langle M_z , (I_d/d + E)^{\otimes t} \rangle \,\mathrm{d}z
\]
where the expectation is over the randomness of the quantum measurement in Algorithm~\ref{alg:learn-balanced}.  We can make the estimator $D_j$ have trace $0$ by simply subtracting out $I_d/d$ and adding it back at the end.  Thus, by Lemma~\ref{lem:estimator-error} and Corollary~\ref{coro:fake-estimator-mean}, recalling the definition of $\theta$ in Line~\ref{step:theta} of Algorithm~\ref{alg:learn-balanced}, we have
\[
\begin{split}
\norm{\E\left[ \frac{D_1 + \dots + D_m}{m} \right] - \frac{I_d}{d} - \frac{d\theta E}{t(d^2 - 1)} }_{F} & \leq \frac{10^5 t^2 \norm{E}_F^2}{d} \sqrt{\int_{\calZ} \frac{\norm{f(M_z)}_F^2 \tr(M_z)}{d^t}\,\mathrm{d}z} \\ & \leq \frac{10^5t^2\norm{E}_F^2}{d} \,.
\end{split}
\]
Thus, if $\wh{E}$ is the output of Algorithm~\ref{alg:learn-balanced}, then
\[
\norm{\E[\wh{E}] - E}_F \leq \frac{10^5t^3 \norm{E}_F^2}{\theta} \,.
\]
Next, we compute the variance of the estimator $\wh{E}$.  We have
\[
\E\left[ \norm{\wh{E} - \E[\wh{E}]}_F^2\right] \leq \frac{d^2t^2}{m\theta^2}  \E\left[ \norm{D_1 - \E[D_1]}_F^2 \right]  \leq \frac{d^2t^2}{m\theta^2}\E\left[\norm{D_1}_F^2 \right] =  \frac{d^2}{m\theta^2}\E_{\lambda \sim \mathrm{SW}^t(\rho)}\Bigl[\sum_{j}\lambda_j^2 \Bigr] \,.
\]
Now by Claim~\ref{claim:typical-young-tableaux2}, we can upper bound $\E_{\lambda \sim \mathrm{SW}^t(\rho)}[\sum_{j}\lambda_j^2 ] \leq 2(\norm{\rho}_F^2t^2 + t^{1.5}) \leq 4t^{1.5}$, where in the last step we used the assumption that $\rho = \frac{I_d}{d} + E$ for $\norm{E}_F \le (0.01/t)^4$ and that $t \le 0.01 d^2$. 
 While Theorem~\ref{thm:learn-balanced} is technically stated for $t \leq d^2$, for $t$ in the range $ 0.01d^2 \leq t \leq d^2$, we can just use $0.01d^2$-entangled measurements instead and this loses at most a constant factor in the total copy complexity.  Also by Claim~\ref{claim:typical-young-tableaux1}, we have $\theta \geq t^{1.5}/4$.  Thus, putting everything together, we conclude
\[
\E\left[ \norm{\wh{E} - E}_F^2\right] \leq 2 \cdot 10^{10} t^3 \norm{E}_F^4 + \frac{(10 d)^2}{m t^{1.5}} \,.
\]
\end{proof}

\section{Reducing to balanced states via quantum splitting}\label{sec:embedding}

Now we demonstrate how to generalize our results in Section~\ref{sec:learn-balanced} to learn arbitrary states.  The main idea will be to take an arbitrary state $\rho$ and construct a state $\Split(\rho)$ that preserves information about $\rho$ and we can simulate measurement access to, but also has bounded operator norm.

\subsection{Construction of the quantum splitting}

We formalize the splitting procedure below.

\begin{definition}\label{def:split}
Let $b_1, \dots , b_d \in \mathbb{Z}_{\geq 0}$.  We define $\Split_{b_1, \dots , b_d}$ to be a linear map that sends any $M\in \C^{d\times d}$ to a square matrix with dimension $2^{d_1} + \dots + 2^{b_d}$ defined as follows.  The rows and columns of $\Split_{b_1, \dots , b_d}(M)$ are indexed by pairs $(j, s)$ where $j \in [d]$ and $s \in \{0,1 \}^{b_j}$ and these are sorted first by $j$ and then lexicographically according to $s$.  Now the entry indexed by row $(j_1, s_1)$ and column $ (j_2,s_2)$  is defined as 
\begin{itemize}
\item If $b_{j_1} \leq b_{j_2}$ then the entry is $M_{j_1j_2}/2^{b_{j_2}}$ if $s_1$ is  a prefix of $s_2$ and is  $0$ otherwise
\item If $b_{j_1} > b_{j_2}$ then the entry is $M_{j_1j_2}/2^{b_{j_1}}$ if $s_2$ is a prefix of $s_1$ and is $0$ otherwise
\end{itemize}
\end{definition}

\noindent Now we can define the natural inverse map to $\Split$.

\begin{definition}
Given $b_1, \dots , b_d \in \mathbb{Z}_{\geq 0}$ and a matrix $N \in  \C^{k \times k}$ where $k = 2^{b_1} + \dots + 2^{b_d}$, we define $\Rec_{b_1, \dots , b_d}(N)$ to be a $d \times d$ matrix $X$ obtained as follows.  First index the rows of $N$ by pairs $(j, s)$ where $j \in [d]$ and $s \in \{0,1 \}^{b_j}$ in sorted order (first by $j$ and lexicographically by $s$).  Now the entry $X_{j_1j_2}$ is equal to
\begin{itemize}
\item If $b_{j_1} \leq b_{j_2}$ then it is $\sum_{s \in \{0,1 \}^{b_{j_2}}} N_{(j_1, s[:b_{j_1}] ), ( j_2, s)}$ where   $s[:b_{j_1}]$ denotes truncating $s$ to its first $b_{j_1}$ bits.  
\item If $b_{j_1} > b_{j_2}$ then it is $\sum_{s \in \{0,1 \}^{b_{j_1}}} N_{(j_1, s ), ( j_2, s[:b_{j_2}])}$ 
\end{itemize}
\end{definition}

\noindent The following basic facts about the splitting and its inverse are easily verified through direct computation.

\begin{claim}\label{claim:basic-embedding}
Let $b_1, \dots , b_d \in \mathbb{Z}_{\geq 0}$. We have the following statements (for any matrices $M,N$ of appropriate dimensions):
\begin{itemize}
\item $\Rec_{b_1, \dots , b_d}(\Split_{b_1, \dots , b_d}(M) ) = M$
\item $\norm{\Split_{b_1, \dots , b_d}(M)}_F \leq \norm{M}_F $
\item $\norm{\Rec_{b_1, \dots , b_d}(N)}_F \leq 2^{\max(b_1, \dots , b_d)/2} \norm{N}_F $
\end{itemize}

\end{claim}
\begin{proof}
The first statement is immediate from the definition since $\Rec$ sums up exactly the entries that are equal to $M_{j_1j_2}/2^{\max(b_{j_1}, b_{j_2})}$ in $\Split$.  The second statement holds because $\Split$ simply splits up each of the entries of $M_{j_1j_2}$ evenly into multiple entries which can only decrease the Frobenius norm.  The last statement holds because each element of $\Rec_{b_1, \dots , b_d}(N)$ is equal to a sum of at most $2^{\max(b_1, \dots , b_d)}$ elements of $N$.  
\end{proof}

Now the key fact about the splitting that makes it a useful abstraction in our learning algorithm is that we can actually simulate measurements on $t$-entangled copies of $\Split_{b_1, \dots , b_d}(\rho)$ with measurements on $t$-entangled copies of $\rho$.

\begin{claim}\label{claim:simulate-measurements}
Given measurement access to $\rho^{\otimes t}$ where $\rho \in \C^{d \times d}$ is a state, $\Split_{b_1, \dots , b_d}(\rho)$ is a valid state and we can simulate measurement access to access to $\Split_{b_1, \dots , b_d}(\rho)^{\otimes t}$.    
\end{claim}
\begin{proof}
Note that $\Split_{b_1, \dots , b_d}(\rho)$ can be constructed by embedding $\rho$ in various different principal submatrices of a $k \times k$ matrix (where $k = 2^{b_1} + \dots + 2^{b_d}$) and averaging them.  In particular, for a string $s \in \{0,1 \}^{\max(b_1, \dots , b_d)}$, we can imagine indexing the rows and columns of the $k \times k$ matrix as in Definition~\ref{def:split} and embedding $\rho$ in the rows and columns indexed by $(1, s[:b_1]), \dots , (d, s[:b_d])$ where $s[:b_j]$ denotes truncating $s$ to its first $b_j$ bits.  Now averaging these embeddings over all $2^{\max(b_1, \dots , b_d)}$ choices for $s$ gives exactly $\Split_{b_1, \dots , b_d}(\rho)$.

Thus, $\Split_{b_1, \dots , b_d}(\rho)^{\otimes t}$ can be constructed by embedding $\rho^{\otimes t}$ in various different principal submatrices of a $k^t \times k^t$ matrix and averaging them.  Thus when we measure $\Split_{b_1, \dots , b_d}(\rho)^{\otimes t}$ with some POVM in $\C^{k^t \times k^t}$, this is equivalent to averaging the different embeddings of $\rho^{\otimes t}$ and thus we can actually simulate this measurement by measuring $\rho^{\otimes t}$ with a single POVM.  
\end{proof}

\subsection{Full algorithm}

In this section, we present our full learning algorithm.  Note that the quantum splitting procedure requires knowledge of $\rho$, or at least an estimate of $\rho$, to be useful.  We can obtain such a rough estimate for $\rho$ via tomography with unentangled measurements.  First, we show how to learn when we are given this estimate as a black-box.  We will then put everything together to prove our main learning result, Theorem~\ref{thm:learn-general}. 

\begin{lemma}\label{lem:learn-bounded}
Let $d,t, \eps, \delta,C $ be parameters and $\rho \in \C^{d \times d}$ be an unknown quantum state.   Assume that $t \leq \min(d^2, (1/\eps)^{0.2})$.  Let $\rho'$ be a quantum state whose description we know such that $\norm{\rho'} \leq C/d$ and $\norm{\rho' - \rho}_F \leq \sqrt{\eps}/t^2$.  Then there is an algorithm that takes  $\wt{O}(d^2 \log(1/\delta)/(\sqrt{t}\eps^2))$ copies of $\rho$ and measures them in batches of $t$-entangled copies  and with probability $1 - \delta$, outputs a state $\wh{\rho}$ such that $\norm{\rho - \wh{\rho}}_F \leq \sqrt{C}\eps$.
\end{lemma}
\begin{proof}
WLOG $\rho'$ is diagonal with eigenvalues $\lambda_1 \geq  \dots  \geq \lambda_d$ (otherwise, we apply the unitary that diagonalzies it and work in this basis instead).  For each $j \in [d]$ let $b_j$ be the smallest nonnegative integer such that $2^{b_j} \geq d\lambda_j$.  Note that we must have $2^{b_1} + \dots + 2^{b_d} \leq 4d$.  Also, $\Split_{b_1,\dots , b_d}(\rho')$ is a diagonal matrix with all entries at most $1/d$ so $\norm{\Split_{b_1,\dots , b_d}(\rho')} \leq 1/d$.  By Claim~\ref{claim:basic-embedding} and linearity of $\Split$,
\[
\norm{\Split_{b_1,\dots , b_d}(\rho') - \Split_{b_1,\dots , b_d}(\rho)}_F \leq \norm{\rho - \rho'}_F \leq \frac{\sqrt{\eps}}{t^2} \,.
\]
Now, by Corollary~\ref{coro:learn-balanced} (with $d$ replaced with $4d$), given $O(d^2 \log(1/\delta)/(\sqrt{t}\eps^2))$ copies of the state $\Split_{b_1,\dots , b_d}(\rho)$, we can make measurements entangled over at most $t$ copies and learn a state $\sigma$ such that 
\[
\norm{\sigma - \Split_{b_1,\dots , b_d}(\rho)}_F \leq 0.1\eps \,.
\]
However by Claim~\ref{claim:simulate-measurements}, we can simply simulate these measurements with copies of $\rho$ instead.  Finally, once we obtain $\sigma$, we simply output $\wh{\rho} = \Rec_{b_1,\dots , b_d}(\sigma)$.  By Claim~\ref{claim:basic-embedding},
\[
\norm{\wh{\rho} - \rho}_F \leq 2^{\max(b_1,\dots , b_d)/2} (0.1 \eps) \leq \sqrt{C} \eps 
\]
by construction, so we are done.
\end{proof}

Sometimes the states that we work with will have large eigenvalues and we won't want to apply Lemma~\ref{lem:learn-bounded} directly as the $\sqrt{C}$ factor loss in the accuracy could be very large.  Instead, we will try to project out the large eigenvalues and learn in the orthogonal complement.  Thus, we have the following subroutine.

 \begin{corollary}\label{coro:learn-projector}
Let $d,t, \eps, \delta,C $ be parameters and $\rho \in \C^{d \times d}$ be an unknown quantum state.  Assume that $C \geq 1$ and $t \leq \min(d^2, (1/\eps)^{0.2})$.  Let $\rho'$ be a quantum state and $P \in \C^{d \times d}$ be a projection matrix whose descriptions we know. Assume that $\tr(P\rho P^\dagger) \geq 0.1$, $\norm{P\rho'P^{\dagger}} \leq C/d$ and $\norm{P(\rho' - \rho)P^\dagger}_F \leq \sqrt{\eps}/t^2$.  Then there is an algorithm that takes  $\wt{O}(d^2 \log(1/\delta)/(\sqrt{t}\eps^2))$ copies of $\rho$, measures them in batches of $t$-entangled copies, and with probability $1 - \delta$ outputs a PSD matrix $\wh{\rho}$ such that $\norm{P(\rho - \wh{\rho})P^\dagger}_F \leq O(\sqrt{C}\eps)$.     
\end{corollary}
\begin{proof}
First estimate $\tr(P \rho P^{\dagger})$ by making (unentangled) measurements on $O(\log(1/\delta)/\eps^2)$ copies of $\rho$ with the POVM $(P, I_d - P)$ and letting $\beta$ be the fraction with outcome $P$.  With probability $1 - 0.1\delta$, we have that $|\beta - \tr(P\rho P^\dagger)| \leq 0.01\eps$ so in particular $\beta \geq 0.09$.

Now we can assume $t \geq 100$ as otherwise we can simply run the algorithm in Theorem~\ref{thm:untentangled-tomography} which uses unentangled measurements.  Now set $t' = t\beta/2 $ and $\eps' = \eps/\beta$.  We apply Lemma~\ref{lem:learn-bounded} with parameters $t',\eps'$ to learn the state 
\[
\rho_0 = \frac{P\rho P^\dagger}{\tr(P\rho P^\dagger)} \,.
\]
To do this, for each copy, of $\rho$, we first measure with the POVM $(P, I_d - P)$ and keep only those whose outcome is $P$.  Now we can simulate measurement access to the state $\rho_0$.  Note that if we have $t$-entangled copies of $\rho$, then with probability at least $1/2$, when we measure them with $(P, I_d - P)$, at least $t'$ of them will have outcome $P$ so then we can make an entangled measurement on $t'$ copies of  $\rho_0$.  In other words, if we have batches of the form $\rho^{\otimes t}$, then with at least half of the batches, we will be able to simulate measurements of $\rho_0^{\otimes t'}$.

By Lemma~\ref{lem:learn-bounded}, with probability $1 - 0.1\delta$, we obtain a state $\wt{\rho}$ such that 
\[
\norm{\wt{\rho} - \frac{P\rho P^\dagger}{\tr(P\rho P^\dagger)} }_F \leq  10^3 \sqrt{C}\eps \,.
\]
Clearly, we can ensure that the state $\wt{\rho}$ lives entirely in the subspace given by $P$.  Now we simply output $\wh{\rho} = \beta \wt{\rho}$.  We have that 
\[
\norm{P(\rho - \wh{\rho})P^\dagger}_F = \norm{\wh{\rho} - P\rho P^\dagger}_F \leq \norm{(\beta - \tr(P\rho P^\dagger))\wt{\rho}}_F +   \norm{\tr(P\rho P^\dagger)\wt{\rho} - P\rho P^\dagger}_F \leq O(\sqrt{C}\eps)
\]
so we are done.
\end{proof}

Now we present our full learning algorithm.  At a high level, we first replace $\rho$ with $\sigma = (\rho + I_d/d)/2$ and learn a rough estimate $\wh{\sigma}_0$ via unentangled measurements.  Then we restrict to subspaces corresponding to eigenvalues of $\wh{\sigma}_0$ at various scales and apply Corollary~\ref{coro:learn-projector} to refine our estimate on each of these subspaces.  Finally, we aggregate our estimates to obtain a refined estimate $\wh{\sigma}$, from which we can recover an estimate of $\rho$.  Note that previously $\eps$ was used to measure the accuracy in Frobenius norm but in Algorithm~\ref{alg:learn-general}, it will be used for accuracy in trace norm.  Our algorithm only guarantees recovery to trace norm $\eps$ with $d^3/(\sqrt{t}\eps^2)$ samples (instead of a stronger Frobenius norm guarantee of $\eps/\sqrt{d}$ which is what we would get in the balanced case) due to the $\sqrt{C}$ factor loss in Corollary~\ref{coro:learn-projector}.  

\begin{algorithm2e}[ht!]\label{alg:learn-general}
\caption{Full Learning Algorithm}
\DontPrintSemicolon
\textbf{Input:} Parameters $d,t,\eps , \delta$ \;
\textbf{Input:} copies of $\rho$ for some unknown quantum state $\rho \in \C^{d \times d}$
\;
Set $n = \wt{O}\left(\frac{d^3 \log(1/\delta)}{\sqrt{t}\eps^2}\right)$  \;
Set $\sigma = (\rho  + I_d/d)/2$ in future steps, simulate measurements of $\sigma$ using the given copies of $\rho$\;
 
Perform tomography with unentangled measurements on $n$ copies of $\sigma$ (Theorem~\ref{thm:untentangled-tomography}) to learn estimate $\wh{\sigma}_0$ for $\sigma$ \;

Let $U$ be the unitary that diagonalizes $\wh{\sigma}_0$ so that 
\[U\wh{\sigma}_0 U^\dagger = \diag(\lambda_1, \dots , \lambda_d)\]
with $\lambda_1 \geq \dots \geq \lambda_d$ \;

\For{$j =1, \dots , \lfloor \log \sqrt{t} \rfloor  $}{
Let $P_j $ be the projection onto the eigenspace corresponding to $\lambda_i \leq \sqrt{t}/(2^{j-1} d)$ \;
Take $n$ copies of $\sigma$ and apply Corollary~\ref{coro:learn-projector} with $\rho' \leftarrow \sigma_0, P \leftarrow P_j, C =\sqrt{t}/2^{j-1}, \eps \leftarrow \eps/\sqrt{d} $ to obtain estimate $\wh{\sigma}_j$ \; 
}
Set $\wh{\sigma} = \sum_{j = 0}^{\lfloor \log \sqrt{t} \rfloor } (P_{j}\wh{\sigma}_jP_j^\dagger - P_{j+1}\wh{\sigma}_jP_{j+1}^\dagger)$ where we define $P_0 = I_d$ and $P_{\lfloor \log \sqrt{t} \rfloor + 1} = 0$\; 

Set $\wh{\rho} = \frac{\trunc(2\wh{\sigma} - I_d/d)}{\tr(\trunc(2\wh{\sigma} - I_d/d))}$ where $\trunc$ denotes zeroing out the negative eigenvalues\;

\textbf{Output:} $\wh{\rho}$\;

\end{algorithm2e}

\noindent Our main theorem is stated below.
\begin{theorem}\label{thm:learn-general}
Let $d,t, \eps, \delta$ be parameters and $\rho \in \C^{d \times d}$ be an unknown quantum state.  If $t \leq \min(d^2 , (\sqrt{d}/\eps)^{0.2})$, then Algorithm~\ref{alg:learn-general} takes  $\wt{O}(d^3 \log(1/\delta)/(\sqrt{t}\eps^2))$ copies of $\rho$ and measures them in batches of $t$-entangled copies and with probability $1 - \delta$, outputs a state $\wh{\rho}$ such that $\norm{\rho - \wh{\rho}}_1 \leq \eps $.
\end{theorem}
\begin{proof}
By Theorem~\ref{thm:untentangled-tomography}, with $1 - 0.1\delta$ probability, we have
\[
\norm{\sigma - \wh{\sigma}_0}_F \leq \frac{\eps t^{0.25}}{\sqrt{d}} \,.
\]
Now we verify that the conditions of Corollary~\ref{coro:learn-projector} hold whenever we apply it in Algorithm~\ref{alg:learn-general}.  The condition on $t$ clearly holds.  Next, note that all eigenvalues of $\sigma$ are at least $1/(2d)$.  For all $j \geq 1$, the dimension of the orthogonal complement of $P_j$ is at most $d/2$ and thus, $\tr(P_j \sigma P_j^\dagger) \geq \frac{1}{4}$.  Also, by definition, we have $\norm{P_j \wh{\sigma}_0 P_j^\dagger} \leq \sqrt{t}/(2^{j-1}d)$.  Finally,
\[
\norm{P(\wh{\sigma}_0 - \sigma) P^\dagger}_F \leq \norm{\sigma - \wh{\sigma}_0 }_F \leq \frac{\eps t^{0.25}}{\sqrt{d}} \leq \frac{\sqrt{\eps/\sqrt{d}}}{t^2}
\]
so we can ensure that with probability $1 - 0.1\delta/t$,
\[
\norm{P_j(\wh{\sigma}_j - \sigma)P_j^\dagger }_F \leq O\left( \frac{t^{0.25}\eps}{2^{j/2}\sqrt{d}}\right) \,.
\]
Thus, since $P_{j+1}$ is a projector onto a subspace of $P_j$, we also have that
\[
\norm{P_j(\wh{\sigma}_j  - \sigma)P_j^\dagger - P_{j+1}(\wh{\sigma}_j  - \sigma) P_{j+1}^\dagger }_F \leq O\left( \frac{t^{0.25}\eps}{2^{j/2}\sqrt{d}}\right) \,.
\]
Next, the dimension of the orthogonal complement of $P_{j+1}$ is at most $2^{j + 1} d/\sqrt{t}$, so
\[
\norm{P_j(\wh{\sigma}_j  - \sigma)P_j^\dagger - P_{j+1}(\wh{\sigma}_j  - \sigma) P_{j+1}^\dagger }_1 \leq O(\eps) \,.
\]
Also, similarly, 
\[
\norm{P_0(\wh{\sigma}_0 - \sigma)P_0^\dagger - P_1 (\wh{\sigma}_0 - \sigma) P_1^\dagger  }_1 \leq O(\eps) \,.
\]
Putting these together, we get that $\norm{\sigma - \wh{\sigma}}_1 \leq O(\eps \log t)$.  This immediately implies 
\[
\norm{(2\wh{\sigma} - I_d/d) - \rho }_1 \leq  O(\eps \log t) \,.
\]
Finally, the truncation and trace normalization to obtain $\wh{\rho}$ increases the trace norm distance by at most a constant factor so we conclude $\norm{\rho - \wh{\rho}}_1 \leq O(\eps \log t)$.  This completes the proof (note that we can simply redefine $\eps$ appropriately and absorb the logarithmic factors in the number of copies into the $\wt{O}( \cdot )$).
\end{proof}

\section{Lower bound machinery}\label{sec:lower-bound-tools}

Recall Definition~\ref{def:unfolded-matrix}.  The key ingredients in our lower bound involve understanding $G_j(v)$, particularly $G_1(v)$, in terms of properties of the vector $v \in \C^{d^t}$. 
 The next lemma relates $\norm{G_1(v)}_F^2$ to the projections of $v$ onto various Schur subspaces of $\C^{d^t}$.

\begin{lemma}\label{lem:skewness-bound}
Let $v \in \C^{d^t}$ be a vector with $\norm{v} = 1$.  Then
\[
\norm{G_1(v)}_F^2 \leq \sum_{\lambda \vdash t} \norm{\Pi_{\lambda} v}^2 ( \lambda_1^2 + \dots + \lambda_d^2) \,.
\]
\end{lemma}
\begin{proof}
Let $A \in \C^{d \times d}$ be an arbitrary Hermitian matrix with $\norm{A}_F = 1$.  Let its eigenvectors be $v_1, \dots , v_d$.  Now we write $v$ in the basis given by $v_{s_1} \otimes \dots \otimes v_{s_t}$ i.e.
\[
v = \sum_{s \in [d]^t} c_s (v_{s_1} \otimes \dots \otimes v_{s_t} ) \,.
\]
For each $s \in [d]^t$ and $k \in [d]$, let $f_s(k)$ denote the number of occurrences of $k$ in $s$.  Note that 
\begin{equation}\label{eq:skewness-bound-step1}
\langle G_1(v), A \rangle = \left\langle \sum_{s \in [d]^t} c_s^2 G_1( v_{s_1} \otimes \dots \otimes v_{s_t} ) , A \right\rangle = \left\langle \sum_{s \in [d]^t} c_s^2 \left(\sum_{k = 1}^d f_s(k) v_kv_k^\dagger \right) , A\right\rangle
\end{equation}
where the first step holds because $A$ is diagonal in the basis $v_1, \dots , v_d$  and all of the cross terms that appear when we expand $G_1(v)$ are off-diagonal.  Now the above is at most 
\begin{equation}\label{eq:skewness-bound-step2}
\begin{split}
\norm{\sum_{s \in [d]^t} c_s^2 \left(\sum_{k = 1}^d f_s(k) v_kv_k^\dagger \right)}_F &\leq \sqrt{\sum_{s \in [d]^t} c_s^2} \cdot \sqrt{ \sum_{s \in [d]^t} c_s^2 \norm{\sum_{k = 1}^d f_s(k) v_kv_k^\dagger}_F^2 } \\ & = \sqrt{ \sum_{s \in [d]^t} c_s^2 \left(\sum_{k = 1}^d f_s(k)^2\right) } 
\end{split}
\end{equation}
since by assumption, $\norm{v}^2 = \sum_{s \in [d]^t} c_s^2 = 1$. 
 For each $s$, let $f_s$ denote the partition corresponding to $(f_s(1), \dots , f_s(d) )$ (in sorted order). We have that
\[
\sum_{\lambda \vdash t} \norm{\Pi_{\lambda} v}^2 \cdot ( \lambda_1^2 + \dots + \lambda_d^2) = \sum_{k = 1}^{t^2} \left(  \sum_{\substack{\lambda \vdash t \\ \lambda_1^2 + \dots + \lambda_d^2 \geq k}} \norm{\Pi_{\lambda} v}^2 \right) \geq  \sum_{k = 1}^{t^2}  \left( \sum_{\substack{s \in [d]^t \\ f_s(1)^2 + \dots  + f_s(d)^2 \geq k}}  c_s^2\right) 
\]
where the last inequality holds because by Lemma~\ref{lem:weight-vector2}, we know that for any $s$ with $f_s(1)^2 + \dots  + f_s(d)^2 \geq k$, the vector $v_{s_1} \otimes \dots \otimes v_{s_t}$ is contained in the space 
\[
\bigoplus_{\substack{\lambda \vdash t \\ \lambda \succeq f_s}} \Pi_{\lambda} \subseteq \bigoplus_{\substack{\lambda \vdash t \\ \lambda_1^2 + \dots + \lambda_d^2 \geq k}} \Pi_{\lambda} \,.
\]
However, we have that
\[
\sum_{k = 1}^{t^2}  \left( \sum_{\substack{s \in [d]^t \\ f_s(1)^2 + \dots  + f_s(d)^2 \geq k}}  c_s^2\right) = \sum_{s \in [d]^t} c_s^2 \left(\sum_{k = 1}^d f_s(k)^2\right) 
\]
so we get
\begin{equation}\label{eq:skewness-bound-step3}
\sum_{\lambda \vdash t} \norm{\Pi_{\lambda} v}^2 \cdot ( \lambda_1^2 + \dots + \lambda_d^2)   \geq  \sum_{s \in [d]^t} c_s^2 \left(\sum_{k = 1}^d f_s(k)^2\right)  \,.
\end{equation}
Putting \eqref{eq:skewness-bound-step1} \eqref{eq:skewness-bound-step2}, \eqref{eq:skewness-bound-step3} together, and taking the maximum over all choices of $A$, we conclude
\[
\norm{G_1(v)}_F^2 \leq \sum_{\lambda \vdash t} \norm{\Pi_{\lambda} v}^2 ( \lambda_1^2 + \dots + \lambda_d^2)
\]
as desired.
\end{proof}

Now we will use Lemma~\ref{lem:skewness-bound} to bound the ``likelihood ratio" $\frac{x^\dagger \rho^{\otimes t}x }{x^\dagger \rho_0^{\otimes t} x}$ for different quantum states $\rho$ and $\rho_0$ and vectors $x \in \C^{d^t}$.  We explain the lower bound framework in detail and why this quantity is meaningful in Section~\ref{sec:lower-bound-proof}.

\begin{lemma}\label{lem:pointwise-ratio-bound}
Let $0 < \eps < 1$ be some parameter.  Let $\rho_0 = (I_d + Z)/d \in \C^{d \times d}$ be a quantum state with $\norm{Z} \leq \eps$.   Let $\mu$ be a distribution on matrices in $\C^{d \times d}$ that is rotationally symmetric i.e. invariant under rotation by any unitary $U$.  Also assume that any $\Delta$ in the support of $\mu$ has $\tr(\Delta) = 0$ and $\norm{\Delta} \leq \eps/d$.  Let $t$ be an integer with $t \leq 0.01/\eps^{0.1}$.  Then for any vector $x \in \C^{d^t}$ with $\norm{x} = 1$,
\[
\E_{\Delta \sim \mu}\left[\log \frac{x^\dagger (\rho_0 + \Delta)^{\otimes t}x }{x^\dagger \rho_0^{\otimes t} x} \right] \geq  \frac{ \frac{1}{d^{t-2}} \E_{\Delta \sim \mu}\left[\langle G_2(x) , \Delta \otimes \Delta \rangle \right] }{x^\dagger \rho_0^{\otimes t} x} - \frac{\frac{2}{d^{t-2}}\E_{\Delta \sim \mu}\left[ \langle G_1(x), \Delta \rangle^2 \right]}{ x^\dagger \rho_0^{\otimes t} x} - \frac{30t^3 \eps^3}{d} \,.
\]
\end{lemma}
\begin{proof}
Using the condition on $t$, we have that all eigenvalues of $\rho_0^{\otimes t}$ and $(\rho_0 + \Delta)^{\otimes t}$ are between $0.98/d^t$ and $1.02/d^t$ so $\frac{x^\dagger (\rho_0 + \Delta)^{\otimes t}x }{x^\dagger \rho_0^{\otimes t} x}  \in [0.9, 1.1]$.  Thus,
\begin{equation}\label{eq:log-expansion}
\log \frac{x^\dagger (\rho_0 + \Delta)^{\otimes t}x }{x^\dagger \rho_0^{\otimes t} x}  \geq \frac{x^\dagger( (\rho_0 + \Delta)^{\otimes t} -  \rho_0^{\otimes t}) x}{x^\dagger \rho_0^{\otimes t} x}  - \frac{2}{3} \left( \frac{x^\dagger((\rho_0 + \Delta)^{\otimes t} -  \rho_0^{\otimes t}) x}{x^\dagger \rho_0^{\otimes t} x}\right)^2 \,.
\end{equation}
For the first term, we can write
\begin{equation}\label{eq:full-expansion}
\begin{split}
(\rho_0 + \Delta)^{\otimes t} -  \rho_0^{\otimes t} & = \sum_{j = 1}^t \sum_{\mathrm{sym}} \Delta^{\otimes j} \otimes \rho_0^{\otimes t - j}  \\ &= \sum_{\mathrm{sym}} \Delta \otimes \rho_0^{\otimes t - 1} + \sum_{j = 2}^t \sum_{k = 0}^{t - j} \sum_{\mathrm{sym}}\Delta^{\otimes j} \otimes (Z/d)^{\otimes k} \otimes (I_d/d)^{\otimes t - j - k } \,.
\end{split}
\end{equation}
Let
\[
T(\Delta) = \sum_{\mathrm{sym}} \Delta^{\otimes 2} \otimes (I_d/d)^{\otimes t - 2} \,.
\]
Since $\mu$ is supported on traceless matrices and is rotationally invariant, $\E_{\Delta \sim \mu}[\Delta] = 0$. Now by Lemma~\ref{lem:random-rotations},
\begin{equation}\label{eq:tensor-norm-bound}
\norm{ \E_{\Delta \sim \mu} [\Delta^{\otimes j}] } \leq \E_{\Delta \sim \mu} \left[ \frac{(4j)^{4j} \norm{\Delta}_F^j}{d^j} \right] \leq \frac{(4j)^{4j} \eps^j}{d^{1.5j}} \,.
\end{equation}
Thus, the first sum in the RHS of \eqref{eq:full-expansion} is $0$ and aside from $T$, the remaining terms all contain a product of at least two copies of $\Delta$ and one copy of $Z$ or at least three copies of $\Delta$.  We can upper bound the norm of the expectation of these terms using \eqref{eq:tensor-norm-bound}.  Combining \eqref{eq:full-expansion}, \eqref{eq:tensor-norm-bound} and using the condition $t \leq 0.01/\eps^{0.1}$, we deduce
\begin{equation}\label{eq:linear-bound}
 \E_{\Delta \sim \mu}[ x^\dagger ( (\rho_0 + \Delta)^{\otimes t} -  \rho_0^{\otimes t}) x  ] \geq x^\dagger\left( \E_{\Delta \sim \mu}\left[ T(\Delta) \right]\right) x - \frac{10t^3\eps^3}{d^{t + 1}} \,.
\end{equation}

Next, we lower bound the quadratic term.  Further expanding \eqref{eq:full-expansion}, we can write
\[
\begin{split}
&(\rho_0 + \Delta)^{\otimes t} -  \rho_0^{\otimes t} \\ &=  \underbrace{\sum_{\mathrm{sym}} \Delta \otimes (I_d/d)^{\otimes t - 1}}_{S_1(\Delta)} + \underbrace{\sum_{k = 1}^{t - 1} \sum_{\mathrm{sym}} \Delta \otimes (Z/d)^{\otimes k} \otimes (I_d/d)^{\otimes t - 1 - k}}_{S_2(\Delta)} +  \underbrace{\sum_{j = 2}^t \sum_{k = 0}^{t - j} \sum_{\mathrm{sym}}\Delta^{\otimes j} \otimes (Z/d)^{\otimes k} \otimes (I_d/d)^{\otimes t - j - k }}_{S_3(\Delta)} 
\end{split}
\]
and we label the three sums above as $S_1(\Delta), S_2(\Delta), S_3(\Delta)$.  Now let us consider expanding $((\rho_0 + \Delta)^{\otimes t} - \rho_0^{\otimes t}) \otimes ((\rho_0 + \Delta)^{\otimes t} - \rho_0^{\otimes t})$.  The key is that in the expansion, aside from $S_1(\Delta) \otimes S_1(\Delta)$, all other terms have a product of at least two copies of $\Delta$ and one copy of $Z$ or at least three copies of $\Delta$ so we can apply \eqref{eq:tensor-norm-bound} to upper bound the norms of all of these terms by $\eps^3/d^{2t+1}$.  Formally, we get
\[
\norm{\E_{\Delta \sim \mu} \left[\left((\rho_0 + \Delta)^{\otimes t} -  \rho_0^{\otimes t}\right) \otimes \left((\rho_0 + \Delta)^{\otimes t} -  \rho_0^{\otimes t} \right) - S_1(\Delta) \otimes S_1(\Delta) \right]} \leq \frac{10t^3\eps^3}{d^{2t+1}} \,.
\]
The above implies that 
\begin{equation}\label{eq:quadratic-bound}
\E_{\Delta \sim \mu}\left[(x^\dagger((\rho_0 + \Delta)^{\otimes t} -  \rho_0^{\otimes t}) x)^2 \right] \leq \E_{\Delta \sim \mu}[ (x^\dagger S_1(\Delta) x)^2] + \frac{10t^3 \eps^3}{d^{2t+1}} \,.
\end{equation}
Finally, observe that by Fact~\ref{fact:unfolded-inner-product}, $x^\dagger S_1(\Delta) x = \frac{1}{d^{t-1}}\langle G_1(x), \Delta \rangle$  and $x^{\dagger} T(\Delta) x = \frac{1}{d^{t-2}} \langle G_2(X) , \Delta \otimes \Delta \rangle$.  Thus, combining \eqref{eq:log-expansion}, \eqref{eq:linear-bound}, \eqref{eq:quadratic-bound}, we get
\[
\E_{\Delta \sim \mu}\left[\log \frac{x^\dagger (\rho_0 + \Delta)^{\otimes t}x }{x^\dagger \rho_0^{\otimes t} x} \right] \geq \frac{ \frac{1}{d^{t-2}} \E_{\Delta \sim \mu}\left[\langle G_2(x) , \Delta \otimes \Delta \rangle \right] }{x^\dagger \rho_0^{\otimes t} x} - \frac{\frac{2}{d^{t-2}}\E_{\Delta \sim \mu}\left[ \langle G_1(x), \Delta \rangle^2 \right]}{ x^\dagger \rho_0^{\otimes t} x} - \frac{30t^3 \eps^3}{d}
\]
where we also used that by the assumption on $t$, $\rho_0^{\otimes t}$ has all eigenvalues at least $0.9/d^t$.  This completes the proof.
\end{proof}

We can aggregate Lemma~\ref{lem:pointwise-ratio-bound} over sequences of vectors $x_1, \dots , x_m$ and apply Jensen's inequality to get the following bound on the product of a sequence of likelihood ratios.

\begin{lemma}\label{lem:avg-likelihood-ratio}
Let $0 < \eps < 1$ be some parameter.  Let $\rho_0 = (I_d + Z)/d \in \C^{d \times d}$ be a quantum state with $\norm{Z} \leq \eps$.   Let $\mu$ be a distribution on matrices in $\C^{d \times d}$ that is rotationally symmetric i.e. invariant under rotation by any unitary $U$.  Also assume that any $\Delta$ in the support of $\mu$ has $\tr(\Delta) = 0$ and $\norm{\Delta} \leq \eps/d$.  Let $t$ be an integer with $t \leq 0.01/\eps^{0.1}$.  Let $x_1, \dots , x_m \in \C^{d^t}$ be vectors such that 
\begin{itemize}
\item $\sum_{\lambda \vdash t} \sum_{j = 1}^m \frac{\norm{\Pi_{\lambda} x_j}^2}{d^t x_j^\dagger \rho_0^{\otimes t} x_j}   (\lambda_1^2 + \dots + \lambda_d^2) \leq A$
\item $\norm{ - \binom{t}{2} m I_{d^2} + \frac{1}{d^{t-2}}\sum_{j = 1}^m \frac{G_2(x_j)}{x_j^\dagger \rho_0^{\otimes} x_j}} \leq B$
\end{itemize}
for some parameter $\beta$.  Then
\[
\E_{\Delta \sim \mu}\left[ \prod_{j = 1}^m \frac{x_j^\dagger (\rho_0 + \Delta)^{\otimes t}x_j}{x_j \rho_0^{\otimes t} x_j }\right] \geq \exp\left( -\frac{2\eps^2 A}{d} - \frac{10^8 \eps^2 B}{d} - \frac{30mt^3\eps^3}{d}\right) \,.
\]
\end{lemma}
\begin{proof}
By Jensen's inequality and Lemma~\ref{lem:pointwise-ratio-bound},
\begin{equation}\label{eq:jensen}
\begin{split}
&\E_{\Delta \sim \mu}\left[ \prod_{j = 1}^m \frac{x_j^\dagger (\rho_0 + \Delta)^{\otimes t}x_j}{x_j \rho_0^{\otimes t} x_j }\right] \geq  \\ & \exp \left(\sum_{j = 1}^m\left(   \frac{ \frac{1}{d^{t-2}} \E_{\Delta \sim \mu}\left[\langle  G_2(x_j) , \Delta \otimes \Delta \rangle \right] }{x_j^\dagger \rho_0^{\otimes t} x_j} - \frac{\frac{2}{d^{t-2}}\E_{\Delta \sim \mu}\left[    \langle G_1(x_j), \Delta \rangle^2 \right]}{ x_j^\dagger \rho_0^{\otimes t} x_j} - \frac{30t^3 \eps^3}{d}\right) \right) \,.
\end{split}
\end{equation}
Now by Lemma~\ref{lem:random-rotations}, $\norm{\E_{\Delta \sim \mu}[ \Delta \otimes \Delta]} \leq \frac{10^{8}\eps^2}{d^3}$ so combining this with the assumption about $x_1, \dots , x_m$ and  that $\tr(\Delta) = 0$, we get
\[
\sum_{j = 1}^m  \frac{ \frac{1}{d^{t-2}} \E_{\Delta \sim \mu}\left[\langle  G_2(x_j) , \Delta \otimes \Delta \rangle \right] }{x_j^\dagger \rho_0^{\otimes t} x_j} \geq -\frac{10^8 \eps^2}{d^3} (B d^2) = \frac{-10^8 \eps^2 B}{d} \,.
\]
Next, by Claim~\ref{claim:haar-moments} and Lemma~\ref{lem:skewness-bound},
\[
\begin{split}
\sum_{j = 1}^m \frac{\frac{2}{d^{t-2}}\E_{\Delta \sim \mu}\left[    \langle G_1(x_j), \Delta \rangle^2 \right]}{ x_j^\dagger \rho_0^{\otimes t} x_j} &\leq \frac{2}{d^{t-2}} \sum_{j = 1}^m \frac{\eps^2}{d^3} \frac{\norm{G_1(x_j)}_F^2}{x_j^\dagger \rho_0^{\otimes t} x_j} \\ &\leq \frac{2\eps^2}{d^{t+1}} \sum_{j = 1}^m \sum_{\lambda \vdash t} \frac{\norm{\Pi_{\lambda} x_j}^2 (\lambda_1^2 + \dots + \lambda_d^2)}{x_j^\dagger \rho_0^{\otimes t} x_j} \\ &\leq \frac{2A\eps^2}{d} \,.
\end{split}
\]
Plugging the above two inequalities into \eqref{eq:jensen}, we get
\[
\E_{\Delta \sim \mu}\left[ \prod_{j = 1}^m \frac{x_j^\dagger (\rho_0 + \Delta)^{\otimes t}x_j}{x_j \rho_0^{\otimes t} x_j }\right] \geq \exp\left( -\frac{2\eps^2 A}{d} - \frac{10^8 \eps^2 B}{d} - \frac{30m t^3\eps^3}{d}\right)
\]
as desired.
\end{proof}

\section{Proof of lower bound}\label{sec:lower-bound-proof}

\subsection{Lower bound framework}
The remainder of the proof of the lower bound will closely follow the framework in \cite{chen2023does} with Lemma~\ref{lem:avg-likelihood-ratio} as the main new ingredient.  Recall that the learner measures $m = n/t$ copies of $\rho^{\otimes t}$ in sequence with POVMs in $\C^{d^t \times d^t}$ possibly chosen adaptively.  It is a standard fact that without loss of generality (see e.g. \cite[Lemma 4.8]{chen2022exponential}) we may assume that all POVMs used are rank-1 and we will work with this assumption for the rest of the lower bound. We will sometimes represent a sequence of $m$ measurement outcomes by $\bx = (x_1,\ldots,x_m)$ to denote that in the $i$-th step, the outcome that was observed corresponds to a POVM element which is a scalar multiple of $x_ix_i^\dagger$.

Next, we review a standard formalism for representing an adaptive algorithm as a tree.  There is a slight difference in the definition here compared to the definition in \cite{chen2022exponential} since we allow the algorithm to make entangled measurements on $t$ copies of $\rho$ simultaneously at each step.

\begin{definition}[Tree representation, see e.g. \cite{chen2022exponential}]\label{def:tree}
    Fix an unknown $d$-dimensional mixed state $\rho$. An algorithm for state tomography that only uses $m$ batches of $t$-entangled copies of $\rho$   can be expressed as a pair $(\cT,\cA)$, where $\cT$ is a rooted tree $\cT$ of depth $m$ satisfying the following properties:
    \begin{itemize}[leftmargin=*,itemsep=0pt]
        \item Each node is labeled by a string of vectors $\bx = (x_1,\ldots,x_k)$, where each $x_i$ corresponds to measurement outcome observed in the $i$-th step.
        \item Each node $\bx$ is associated with a probability $p^{\rho^{\otimes t}}(\bx)$ corresponding to the probability of observing $\bx$ over the course of the algorithm. The probability for the root is 1.
        \item At each non-leaf node, we measure $\rho^{\otimes t}$ using a rank-1 POVM $\lt\{\omega_x d^t \cdot xx^{\dagger}\rt\}_x$ to obtain classical outcome (that is a unit vector) $x\in \C^{d^t}$. The children of $\bx$ consist of all strings $\bx' = (x_1,\ldots,x_k,x)$ for which $x$ is a possible POVM outcome.
        \item If $\bx' = (x_1,\ldots,x_k,x)$ is a child of $\bx$, then
        \begin{equation}
            p^{\rho^{\otimes t}}(\bx') = p^{\rho^{\otimes t}}(\bx)\cdot \omega_x d^t \cdot x^{\dagger} \rho^{\otimes t} x\,.
        \end{equation}
        \item Every root-to-leaf path is length-$m$. Note that $\cT$ and $\rho$ induce a distribution over the leaves of $\cT$.
    \end{itemize}
    $\cA$ is a randomized algorithm that takes as input any leaf $\bx$ of $\cT$ and outputs a state $\cA(\bx)$. The output of $(\cT,\cA)$ upon measuring $m$ copies of a state $\rho$ in $t$-entangled batches is the random variable $\cA(\bx)$, where $\bx$ is sampled from the aforementioned distribution over the leaves of $\cT$.
\end{definition}
\noindent

We also recall the definition of the Gaussian Unitary Ensemble (GUE) and define a trace-centered variant, which will be the basis of our hard distribution.

\begin{definition}[GUE, Trace-centered GUE]
    A sample $G \sim \GUE(d)$ is a Hermitian matrix with independent Gaussians on and above the diagonal, with $G_{j,j} \sim \cN(0,2/d)$ and $G_{j_1,j_2} \sim \cN(0,1/d) + i \cN(0,1/d) $ for $j_1 < j_2$.
    A sample $G' \sim \GUEc(d)$ is sampled by $G' = G - \Tr(G) I_d/d$ where $G \sim \GUE(d)$.
\end{definition}

We recall the following standard fact about extremal eigenvalues of the GUE matrix.

\begin{fact}
    \label{fact:goe-op-norm}
    If $G \sim \GUEc(d)$, then $\norm{G} \le 3$ with probability $1-e^{-\Omega(d)}$.
\end{fact}

\subsection{Construction of hard distribution}

We construct the following hard distribution $\mu$ over quantum states.
Let $U \subseteq  \C^{d\times d}$ be the subspace of Hermitian matrices with trace $1$ and $U_0 \subseteq \C^{d\times d}$ be the subspace of Hermitian matrices with trace $0$.  These spaces inherit the inner product of $\C^{d\times d}$, which defines Lebesgue measures $\Leb_U$ and $\Leb_{U_0}$ on them.  Let $\eps$ be the target accuracy.  We can assume $\eps$ is sufficiently small and let $\sigma = C \eps$ for some constant $C > 1$ to be chosen later.
A sample $\rho \sim \mu$ is generated by
\[
    \rho = \fr1d (I_d + \sigma G),
\]
where $G$ is a sample from $\GUEc(d)$ conditioned on $\norm{G} \le 4$.
Note that such matrices are clearly valid quantum states.
Concretely, $\mu$ has density (with respect to $\Leb_U$)
\[
    \mu(\rho)
    =
    \fr{1}{Z}
    \exp\lt(
        -\fr{d^3}{4\sigma^2}
        \norm{\rho - \fr{1}{d}I_d}_F^2
    \rt)
    \ind\{\rho \in S_{\supp}\},
    \qquad
    S_{\supp} = \lt\{
        \rho \in U : \norm{\rho - \fr{1}{d}I_d} \le \fr{4\sigma}{d}
    \rt\}.
\]
where $Z$ is a normalizing constant.
Further define a set of ``good" states
\[
    S_{\good} = \lt\{
        \rho \in U : \norm{\rho - \fr{1}{d}I_d} \le \fr{3\sigma}{d}
    \rt\},
\]
which corresponds to the event
$\norm{G} \le 3$.
Due to Lemma~\ref{fact:goe-op-norm}, $\mu(S_{\good}) \ge 1-e^{-\Omega(d)}$.

In the below proof, we will show that all $\rho_0 \in S_{\good}$ are hard to learn.
The important property of $S_{\good}$ is that it is far from the boundary of $\supp(\mu) = S_{\supp}$; this  ensures that we can choose a suitable sub-sampling of $\mu$ in a neighborhood of $\rho_0$, which is rotationally symmetric around $\rho_0$.

Finally, we record the following straightforward fact.
\begin{lemma}
    \label{lem:mu-ratio}
    For all $\rho,\rho' \in S_{\supp}$, $\exp(-4d^2) \le \mu(\rho)/\mu(\rho') \le \exp(4d^2)$.
\end{lemma}
\begin{proof}
    For all $\rho \in S_{\supp}$,
    \[
        0\le
        \fr{d^3}{4\sigma^2} \norm{\rho-\fr{1}{d}I_d}_F^2
        \le
        \fr{d^4}{4\sigma^2}\norm{\rho-\fr{1}{d}I_d}^2
        \le
        4d^2. \qedhere
    \]
\end{proof}

\subsection{Anticoncentration of posterior distribution}

Fix a tomography algorithm $(\cT,\cA)$ as in Definition~\ref{def:tree}, and let $\cT_\rho$ denote the distribution over observation sequences $\bx = (x_1,\ldots,x_{m})$ when $\cT$ is run on state $\rho$.
Note that for any states $\rho, \rho'$ in the support of $\mu$, the likelihood ratio 
\[
    \fr{\de \cT_\rho}{\de \cT_{\rho'}}
    (\bx)
    =
    \prod_{i=1}^{m}
    \fr{x_i^\dagger \rho^{\otimes t} x_i}{x_i^\dagger \rho'^{\otimes t} x_i}
\]
is well defined, since $\mu$ is supported on full-rank matrices.
Let $\nu_{\bx}$ denote the posterior distribution of $\rho$ given observations $\bx$.
The density ratio of any $\rho,\rho' \in S_{\supp}$ under $\nu_{\bx}$ is given by Bayes' rule, and equals
\[
    \fr{\nu_{\bx}(\rho)}{\nu_{\bx}(\rho')}
    = \fr{\de \cT_\rho}{\de \cT_{\rho'}} (\bx)
    \cdot \fr{\mu(\rho)}{\mu(\rho')}.
\]
So, for an arbitrary reference state $\rho' \in S_{\supp}$ (below we take $\rho' = \rho_0$, the unknown true state) the density of $\nu_{\bx}$ is
\[
    \nu_{\bx}(\rho)
    =
    \fr{1}{Z_{\bx}}
    \fr{\de \cT_\rho}{\de \cT_{\rho'}} (\bx)
    \mu(\rho),
    \qquad
    Z_{\bx}
    =
    \int_{U}
    \fr{\de \cT_\rho}{\de \cT_{\rho'}} (\bx)
    \mu(\rho)
    ~\de \Leb_U(\rho).
\]
The main technical component of the proof is the following anti-concentration result for $\nu_\bx$.

\begin{definition}
    For $\rho \in U$, let $B(\rho, \eps)$ denote the ball $\{\rho' \in U : \norm{\rho'-\rho}_{1} \le \eps\}$.
    Similarly for $\rho \in U_0$, let $B(\rho, \eps) = \{\rho' \in U_0 : \norm{\rho'-\rho}_{1} \le \eps\}$.
\end{definition}
\begin{theorem}
    \label{thm:main}
    Suppose $d \gg 1$, $\eps \le \eps_0$ for an absolute constant $\eps_0$ and $t$ satisfies $t \leq 1/(C\eps^{0.1})$ for some absolute constant $C$.  Also assume $m \ll d^3/(t^{1.5}\eps^2)$.    If $\rho_0 \in S_{\good}$ and
    $\bx \sim \cT_{\rho_0}$, there is an event $S_{\rho_0} \in \sigma(\bx)$, with $\bbP(\bx \in S_{\rho_0}) \ge 1 - o(1)$, on which $\nu_\bx(B(\rho_0,\eps)) \ll 1$.
\end{theorem}

Let $C$ be a large constant we will set later.  The starting point of the proof of Theorem~\ref{thm:main} is the estimate
\begin{align}
    \label{eq:likelihood-ratio-0}
    \fr{\nu_\bx(B(\rho_0,C\eps))}{\nu_\bx(B(\rho_0,\eps))}
    &=
    \fr{
        \int_{B(\rho_0, C\eps)}
        \fr{\de \cT_\rho}{\de \cT_{\rho_0}}(\bx)
        \mu(\rho)
        ~\de \Leb_U(\rho)
    }{
        \int_{B(\rho_0,\eps)}
        \fr{\de \cT_\rho}{\de \cT_{\rho_0}}(\bx)
        \mu(\rho)
        ~\de \Leb_U(\rho)
    } \\
    &\ge
    \label{eq:likelihood-ratio-1}
    \exp(-4d^2)
    \fr{
        \int_{B(\rho_0, C\eps)}
        \fr{\de \cT_\rho}{\de \cT_{\rho_0}}(\bx)
        \ind\{\rho \in S_{\supp}\}
        ~\de \Leb_U(\rho)
    }{
        \int_{B(\rho_0,\eps)}
        \fr{\de \cT_\rho}{\de \cT_{\rho_0}}(\bx)
        \ind\{\rho \in S_{\supp}\}
        ~\de \Leb_U(\rho)
    },
\end{align}
where the second line uses Lemma~\ref{lem:mu-ratio}.
Applying Lemma~\ref{lem:mu-ratio} in this way amounts to replacing $\mu$ in the numerator of \eqref{eq:likelihood-ratio-0} with a measure that sub-samples it, and in the denominator with a measure that upper bounds it.
Define the volumes
\[
    V_1 = \int_{B(0,1)}~\de \Leb_{U_0}(\rho),
    \qquad
    V_2 = \int_{B(0,1)} \ind\lt\{\norm{\rho} \le 1/d\rt\} ~\de \Leb_{U_0}(\rho).
\]
We now separately bound the numerator and denominator of \eqref{eq:likelihood-ratio-1} in terms of these volumes, beginning with the denominator.
\begin{lemma}
    \label{lem:denominator}
    If $\rho_{0} \in S_{\good}$, there is an event $S_{\rho_0} \in \sigma(\bx)$, with $\Pr(\bx \in S_{\rho_0}) \ge 1 - \exp(-d^2)$, on which
    \[
        \int_{B(\rho_0,\eps)}
        \fr{\de \cT_\rho}{\de \cT_{\rho_0}}(\bx)
        \ind\{\rho \in S_{\supp}\}
        ~\de \Leb_U(\rho)
        \le \exp(d^2) \eps^{d^2 - 1} V_1\,.
    \]
\end{lemma}
\begin{proof}
    Note that $\E_{\bx \sim \cT_{\rho_0}} \fr{\de T_\rho}{\de T_{\rho_0}}(\bx) = 1$.
    So
    \begin{align*}
        \E_{\bx \sim \cT_{\rho_0}}
        \int_{B(\rho_0,\eps)}
        \fr{\de \cT_\rho}{\de \cT_{\rho_0}}(\bx)
        \ind\{\rho \in S_{\supp}\}
        ~\de \Leb_U(\rho)
        &=
        \int_{B(\rho_0,\eps)}
        \ind\{\rho \in S_{\supp}\}
        ~\de \Leb_U(\rho) \\
        &\le
        \int_{B(\rho_0,\eps)}
        ~\de \Leb_U(\rho)
        =
        \eps^{d^2 - 1} V_1.
    \end{align*}
    The exponent $d^2 - 1$ comes from the fact that the space of complex Hermitian matrices has dimension $d^2$, so $U$ has dimension $d^2 - 1$.
    The result follows from Markov's inequality.
\end{proof}
Before bounding the numerator of \eqref{eq:likelihood-ratio-1}, we define the set
\[
    N_\ast(\rho_0) = \lt\{
        \rho \in U :
        \norm{\rho - \rho_0} \le \fr{C\eps}{d},
        \norm{\rho - \rho_0}_{1} \le C\eps
    \rt\},
\]
which is a rotationally symmetric neighborhood of $\rho_0$.  Let $\gamma(\rho_0)$ denote the uniform distribution on $N_\ast(\rho_0)$ (w.r.t. $\Leb_U$).
That is, for bounded measurable test function $f : U \to \bbR$,
\[
    \E_{\rho \sim \gamma(\rho_0)} f(\rho)
    =
    \fr{
      \int_U f(\rho)
      \ind\{\rho \in N_\ast(\rho_0)\}
      \de \Leb_U(\rho)
    }{
      \int_U
      \ind\{\rho \in N_\ast(\rho_0)\}
      \de \Leb_U(\rho)
    }.
\]
\begin{lemma}
    \label{lem:numerator}
    If $\rho_0 \in S_{\good}$ and $\eps \le \eps_0$ for an absolute constant $\eps_0$, then
    \begin{equation}
      \label{eq:numerator-bd}
      \int_{B(\rho_0, C\eps)}
      \fr{\de \cT_\rho}{\de \cT_{\rho_0}}(\bx)
      \ind\{\rho \in S_{\supp}\}
      ~\de \Leb_U(\rho)
      \ge
      (C\eps)^{d^2 - 1}
      \E_{\rho \sim \gamma(\rho_0)}\lt[\fr{\de \cT_\rho}{\de \cT_{\rho_0}}(\bx)\rt]
      V_2.
    \end{equation}
\end{lemma}
\begin{proof}
    Because $\rho_0 \in S_{\good}$, we have $\norm{\rho_0-\fr1d I_d} \le \fr{3\sigma}{d}$.
    Because $\eps \le \sigma/C$, if $\rho$ satisfies $\norm{\rho-\rho_0} \le \fr{C\eps}{d}$ we have the implication chain
    \[
      \norm{\rho-\rho_0} \le \fr{C\eps}{d}
      \Rightarrow
      \norm{\rho-\rho_0} \le \fr{\sigma}{d}
      \Rightarrow
      \norm{\rho-\fr1d I_d} \le \fr{4\sigma}{d}
      \Rightarrow
      \rho \in S_{\supp}.
    \]
    So, letting $X$ denote the left-hand side of \eqref{eq:numerator-bd}, we have
    \begin{align*}
      X &\ge
      \int_{B(\rho_0, C\eps)}
      \fr{\de \cT_\rho}{\de \cT_{\rho_0}}(\bx)
      \ind\lt\{\norm{\rho-\rho_0} \le \fr{C\eps}{d}\rt\}
      ~\de \Leb_U(\rho)
      =
      \int_{U}
      \fr{\de \cT_\rho}{\de \cT_{\rho_0}}(\bx)
      \ind\{\rho \in N_\ast(\rho_0)\}
      ~\de \Leb_U(\rho) \\
      &=
      \E_{\rho \sim \gamma(\rho_0)}
      \lt[\fr{\de \cT_\rho}{\de \cT_{\rho_0}}(\bx)\rt]
      \int_U \ind\{\rho \in N_\ast(\rho_0)\}
      ~\de \Leb_U(\rho).
    \end{align*}
    Finally, note that
    \begin{align*}
      \int_U \ind\{\rho \in N_\ast(\rho_0)\}
      ~\de \Leb_U(\rho)
      =
      \int_{B(\rho_0, C\eps)}
      \ind\lt\{\norm{\rho - \rho_0} \le \fr{C\eps}{d}\rt\}
      ~\de \Leb_U(\rho)
      = (C\eps)^{d^2 - 1} V_2,
    \end{align*}
    which conludes the proof.
\end{proof}

It remains to control $\E_{\rho \sim \gamma(\rho_0)} \lt[\fr{\de \cT_\rho}{\de \cT_{\rho_0}}(\bx)\rt]$ and the volume ratio $V_2/V_1$.  We focus on the former first, the latter can be bounded using formulas from random matrix theory as in  \cite{chen2023does}.  We will need to condition on a ``good" event for the observations $\bx$.  Note that this complication does not show up in \cite{chen2023does} but shows up here due to the more complicated nature of the likelihood ratio expressions with $t$-entangled measurements.  The event that we condition on is defined below. 
\begin{definition}
We say that a sequence of vectors $\bx = (x_1, \dots , x_m)$ where $x_j \in \C^{d^t}$ is well-balanced if
\begin{itemize}
\item $\sum_{\lambda \vdash t} \sum_{j = 1}^m \frac{\norm{\Pi_{\lambda} x_j}^2}{d^t x_j^\dagger \rho_0^{\otimes t} x_j}   (\lambda_1^2 + \dots + \lambda_d^2) \leq 10(mt^{1.5} + \log(d/\eps)\sqrt{m} t^2)$
\item $\norm{ - \binom{t}{2} m I_{d^2} + \frac{1}{d^{t-2}}\sum_{j = 1}^m \frac{G_2(x_j)}{x_j^\dagger \rho_0^{\otimes} x_j}} \leq  10 t^2  d \sqrt{m} \log(d/\eps)$
\end{itemize}
\end{definition}

We first prove that the set of observations $\bx$ is indeed well-balanced with high probability.  
\begin{lemma}\label{lem:good-whp}
Let $\eps, d$ be parameters such that $\eps \leq \eps_0$ for some sufficiently small absolute constant $\eps_0$.  Let $t,m$ be parameters such that $t \leq \frac{1}{10C\eps}$ where $C$ is some constant.  Let $\rho_0 \in \C^{d \times d}$ be a state with $\norm{\rho_0 - I_d/d}  \leq \frac{4C\eps}{d}$.  Measure the state $\rho_0^{\otimes t}$ sequentially $m$ times with arbitrary rank-$1$ POVMs (possibly chosen adaptively) and let the outcomes be $\bx = (x_1, \dots , x_m)$ where $x_j$ is a unit vector in $\C^{d^t}$.  Then with probability $1 - (\eps/d)^{10}$, the collection $\bx$ is well-balanced.
\end{lemma}
\begin{proof}
For the first claim, note that for any POVM $\calM = \{ \omega_x d^t xx^\dagger \}_x$, if we measure $\rho_0^{\otimes t}$ with this POVM, then 
\begin{equation}\label{eq:isotropic-expectation}
\E_{\calM}\left[\frac{x x^\dagger}{d^t x^\dagger \rho_0^{\otimes t}x} \right] = \sum_{x \in \calM} \omega_x xx^\dagger = \frac{I_{d^t}}{d^t} 
\end{equation}
where the expectation is over the outcome from $\calM$.  Thus,
\[
\E_{\calM}\left[\sum_{\lambda \vdash t} \frac{\norm{\Pi_{\lambda} x}^2}{d^t x^\dagger \rho_0^{\otimes t} x}   (\lambda_1^2 + \dots + \lambda_d^2)  \right] = \sum_{\lambda \vdash t} \frac{\dim(\lambda)\dim(V_{\lambda}^d)}{d^t}(\lambda_1^2 + \dots + \lambda_d^2)  = \E_{\lambda \sim SW^t_d}[ \lambda_1^2 + \dots + \lambda_d^2] \leq 3t^{1.5} 
\]
where the last step uses Claim~\ref{claim:typical-young-tableaux2}.  Also note that the quantity inside the expectation is at most $2t^2$ always since $\rho_0$ has all eigenvalues at least $\frac{1}{2d^t}$ by the assumption on $t$.  Thus, since the above holds for any POVM, we can apply Azuma's inequality and get that with probability at least $1 - (\eps/d)^{20}$,
\[
\sum_{\lambda \vdash t} \sum_{j = 1}^m \frac{\norm{\Pi_{\lambda} x_j}^2}{d^t x_j^\dagger \rho_0^{\otimes t} x_j}   (\lambda_1^2 + \dots + \lambda_d^2) \leq 10(mt^{1.5} + \log(d/\eps)\sqrt{m} t^2)
\]
and thus the first condition of well-balanced is satisfied.  Now we analyze the second condition.  By \eqref{eq:isotropic-expectation}, we also have
\[
\E_{\calM}\left[ \frac{1}{d^{t-2}}\frac{G_2(x)}{x^\dagger \rho_0^{\otimes t} x}\right] = \binom{t}{2}I_{d^2} \,.
\]
Also, the quantity inside the expectation is PSD and has trace at most $2t^2d^2$ always.  This means its variance has operator norm at most $2t^4d^2 $.  Since the above holds for any POVMs, we can apply matrix Azuma and get that with probability at least $1 - (\eps/d)^{20}$,
\[
\norm{ - \binom{t}{2} m I_{d^2} + \frac{1}{d^{t-2}}\sum_{j = 1}^m \frac{G_2(x_j)}{x_j^\dagger \rho_0^{\otimes} x_j}} \leq 10 t^2  d \sqrt{m} \log(d/\eps)
\]
which matches the second condition for well-balanced, as desired.
\end{proof}

Now we can bound the average likelihood ratio whenever the sequence of observations is well-balanced.
\begin{lemma}
    \label{lem:likelihood-ratio-improvement}
    Assume that $\eps \leq \eps_0$ and  $ t \leq 1/(C\eps^{0.1})$ for some absolute constants $C, \eps_0$ and $m \ll d^3/(t^{1.5}\eps^2)$.  If $\rho_0 \in S_{\good}$, then for any well-balanced sequence of unit vectors $\bx = (x_1,\ldots,x_n)$,
    \[
        \E_{\rho \sim \gamma(\rho_0)}
        \lt[\fr{\de \cT_\rho}{\de \cT_{\rho_0}}(\bx)\rt]
        \ge
        \exp(-d^2).
    \]
\end{lemma}
\begin{proof}
Note that the distribution $\gamma(\rho_0)$ is rotationally symmetric around $\rho$.  Thus, we can apply Lemma~\ref{lem:avg-likelihood-ratio} with $\eps \leftarrow C\eps$ and
\[
\begin{split}
A &= 10(m t^{1.5} + \log(d/\eps) \sqrt{m} t^2) \\
B &= 10t^2 d \sqrt{m} \log(d/\eps) \,.
\end{split}
\]
Plugging in the assumptions about the parameters $m,t$, we get
 \[
        \E_{\rho \sim \gamma(\rho_0)}
        \lt[\fr{\de \cT_\rho}{\de \cT_{\rho_0}}(\bx)\rt]
        \ge
        \exp(-d^2)
    \]
as desired.
\end{proof}
The volume ratio $V_2/V_1$ is bounded by the following lemma from \cite{chen2023does} \footnote{Technically \cite{chen2023does} works with real symmetric matrices instead of Hermitian matrices but the proofs are exactly the same.}.
\begin{lemma}[See \cite{chen2023does}]
    \label{lem:volume-ratio}
    We have that $V_2/V_1 \ge \exp(-10d^2)$.
\end{lemma}

We can now combine the above lemmas together to prove Theorem~\ref{thm:main}. 

\begin{proof}[Proof of Theorem~\ref{thm:main}]
    Let $S_{\rho_0}$ be the intersection of the event from Lemma~\ref{lem:denominator} and that the observations $\bx$ are well-balanced.  By Lemma~\ref{lem:good-whp}, this event occurs with $1 - o(1)$ probability.  Next, by the calculation \eqref{eq:likelihood-ratio-1} and Lemmas~\ref{lem:denominator} and \ref{lem:numerator}, for all $\bx \in S_{\rho_0}$,
    \[
        \fr{\nu_\bx(B(\rho_0,C\eps))}{\nu_\bx(B(\rho_0,\eps))}
        \ge
        \exp(-5d^2)
        C^{d^2 - 1}
        \E_{\rho \sim \gamma(\rho_0)}
        \lt[\fr{\de \cT_\rho}{\de \cT_{\rho_0}}(\bx)\rt]
        \cdot \fr{V_2}{V_1}.
    \]
    Lemmas~\ref{lem:likelihood-ratio-improvement} and \ref{lem:volume-ratio} bound the remaining factors, giving
    \[
        \fr{\nu_\bx(B(\rho_0,C\eps))}{\nu_\bx(B(\rho_0,\eps))}
        \ge
        \exp(-20d^2)
        C^{d^2 - 1}.
    \]
    Taking $C = e^{30}$ gives $\fr{\nu_\bx(B(\rho_0,C\eps))}{\nu_\bx(B(\rho_0,\eps))} \gg 1$.
    Since $\nu_\bx(B(\rho_0,C\eps)) \le 1$, this implies $\nu_\bx(B(\rho_0,\eps)) \ll 1$.
\end{proof}

\subsection{Completing the lower bound}
We can now complete the proof of our full lower bound for tomography with $t$-entangled measurements, which we state formally below:

\begin{theorem}\label{thm:tomo_lbd}
    There exist absolute constants $\eps_0 > 0$ and $d_0\in\mathbb{N}$ such that for any $0 < \eps < \eps_0$ and any integer $d\ge d_0$ and parameter $t \leq 1/\eps^{0.1}$, the following holds. If $n = o(d^3/(\sqrt{t}\eps^2))$, then for any algorithm for state tomography  $(\cT,\cA)$ that uses $n$ total copies of $\rho$ in (at most) $t$-entangled batches, its output $\widehat{\rho}$ after making all of its measurements satisfies $\norm{\rho - \widehat{\rho}}_{1} > \eps$ with probability $1 - o(1)$.
\end{theorem}

\begin{proof}
First, note that we can reduce to the case where the algorithm always makes exactly $t$-entangled measurements because for any algorithm that makes measurements with lesser entanglement in sequence, we can simply combine them into one larger entangled measurement (this comes at the cost of wasting up to half of the copies of $\rho$ when we make all of the measurements exactly $t$-entangled).

Let $S\in\sigma(\rho,\bx)$ be the event that $\rho\sim\mu$ lies in $S_{\good}$ and $\bx\sim \cT_{\rho}$ lies in $S_{\rho}$. In this proof we will abuse notation and use $\cA$ to also denote the internal randomness used by $\cA$. It suffices to show \[\Pr_{\cA,\rho\sim\mu,\bx\sim\cT_\rho}[\norm{\cA(\bx) - \rho}_{1} \leq \eps] = o(1) \,.
\]

    First note that
    \begin{equation}
    \begin{split}
        \Pr_{\cA,\rho,\bx}\left[\norm{\cA(\bx) - \rho}_{1} \le \eps \right] &= \E_{\cA,\bx}\E_{\rho\sim\nu_{\bx}}\left[\bone{\norm{\cA(\bx) - \rho}_{1} \le \eps}\right] \\
        &\le \E_{\cA,\bx}\E_{\rho\sim\nu_{\bx}}\left[\bone{\norm{\cA(\bx) - \rho}_{1} \le \eps \ \text{and} \ (\rho,\bx)\in S}\right] + o(1) \label{eq:union}
    \end{split}
    \end{equation}
    where the second step follows by a union bound and the fact that $\Pr[(\rho,\bx)\not\in S] = o(1)$ by Theorem~\ref{thm:main}.

    For any choice of internal randomness for $\cA$ and any transcript $\bx$, let $\rho^\cA_\bx$ denote an arbitrary state for which $(\rho^\cA_\bx,\bx)\in S$ and $\norm{\cA(\bx) - \rho^\cA_\bx}_{1} \le \eps$, if such a state exists. Denote by $\cE$ the event that such a state exists. Then under $\cE$, for any state $\rho$ for which $\norm{\cA(\bx) - \rho}_1 \le \eps$, we have $\norm{\rho^\cA_\bx - \rho}_1 \le 2\eps$. If $\cE$ does not occur for some choice of internal randomness for $\cA$ and some $\bx$, note that the corresponding inner expectation in \eqref{eq:union} is zero.
    We can thus upper bound the double expectation in \eqref{eq:union} by
    \begin{align*}
        \E_{\cA,\bx | \cE}\E_{\rho\sim\nu_{\bx}}\left[\bone{\norm{\rho^\cA_\bx-\rho}_{1} \le 2\eps \ \text{and} \ (\rho,\bx)\in S}\right] \le \E_{\cA,\bx | \cE}\Pr_{\rho\sim\nu_{\bx}}\left[\norm{\rho^\cA_\bx-\rho}_{1} \le 2\eps\right] = o(1),
    \end{align*}
    where in the last step we used the fact that under $\cE$ we have $(\rho^\cA_\bx,\bx)\in S$, so by Theorem~\ref{thm:main} the posterior measure $\nu_\bx$ places $o(1)$ mass on the trace norm $\eps$-ball around $\rho^\cA_\bx$.
\end{proof}

\section*{Acknowledgements}

The authors would like to thank Jordan Cotler, Weiyuan Gong, Hsin-Yuan Huang, Ryan O'Donnell, John Wright, and Qi Ye for enlightening conversations at various points about quantum learning with partially entangled measurements.

\bibliographystyle{plain}
\bibliography{bibliography}

\end{document}

%% file: tech-overview.tex
\section{Technical overview}
\label{sec:overview}

\subsection{Basic Setup}

Throughout, let $\rho$ denote a mixed state.
We recall that a mixed state is described by its density matrix, a PSD Hermitian matrix in $\C^{d\times d}$ with trace $1$.  We use $I_d$ to denote the $d \times d$ identity matrix. Given matrix $M$, we use $\norm{M}$, $\norm{M}_1$, and $\norm{M}_F$ to denote its operator norm, trace norm, and Frobenius norm respectively.

\paragraph{Measurements.}  
We now define the standard measurement formalism, which is the way algorithms are allowed to interact with a quantum state $\rho$.
\begin{definition}[Positive operator valued measurement (POVM), see e.g.~\cite{nielsen2002quantum}]
A positive operator valued measurement $\cM$ in $\C^{d \times d}$ is a collection of PSD matrices $\cM = \{M_z\}_{z \in \cZ}$ satisfying $\sum_z M_z = I_d$.  When a state $\rho$ is measured using $\cM$, we get a draw from a classical distribution over $\cZ$, where we observe $z$ with probability $\tr (\rho M_z)$.
\end{definition}

\noindent 

\paragraph{$t$-Entangled Measurements}
An algorithm that uses $t$-entangled measurements and $n$ total copies of $\rho$ operates as follows: it is given $m = n/t$ copies of $\rho^{\otimes t}$.  It then iteratively measures the $i$-th copy of $\rho^{\otimes t}$ (for $i = 1,2, \dots , m$) using a POVM  in $\C^{d^t \times d^t}$ which could depend on the results of previous measurements, records the outcome, and then repeats this process on the $(i + 1)$-th copy of $\rho^{\otimes t}$.  After performing all $m$ measurements, it must output an estimate of $\rho$ based on the (classical) sequence of outcomes it has received.

\begin{remark}
Note that limiting the batch size to exactly $t$ is not actually a limitation because we can simulate an algorithm that adaptively chooses batch sizes of at most $t$ with an algorithm that uses batch sizes of exactly $t$ up to at most a factor of $2$ in the total copy complexity.  This is because we can simply combine batches whenever their total size is at most $t$.   
\end{remark}

\subsection{Key Technical Tools}
Now we introduce some of the key technical tools used in our proofs.  We first focus on the learning algorithm.  We will discuss the proof of the matching lower bound at the end of this section \--- it turns out that the lower bound actually reuses many of the same ideas.

\subsubsection{Power Series Approximation}\label{sec:power-series}

It will be instructive to first consider the case where $\rho$ is close to the maximally mixed state.  Write $\rho = I_d/d + E$ for some $E \in \C^{d \times d}$ with $\tr(E) = 0$.  Observe that when $E$ is small, we can approximate
\begin{equation}\label{eq:intro-powerseries}
\rho^{\otimes t} \approx (I_d/d)^{\otimes t} + \sum_{\mathrm{sym}} E \otimes (I_d/d)^{\otimes t-1} \,.
\end{equation}
Let the LHS of the above be $X$ and the RHS be $X'$.  In general $X'$ will be significantly easier to work with than $X$ when analyzing various measurements so when the above approximation is good, we can work with $X'$ in place of $X$.  Now consider a POVM, say $\{\omega_z zz^\dagger \}_{z \in \calZ}$ in $\C^{d^t \times d^t }$ where $z$ are unit vectors and $\omega_z$ are nonnegative weights (WLOG it suffices to consider rank-$1$ POVMs, see e.g. \cite{chen2022exponential, chen2022tight, chen2023does}).  Now consider constructing an estimator $f: \C^{d^t} \rightarrow \C^{d \times d}$ that is supposed to estimate $E$.  The expectation after measuring $X'$ is 
\begin{equation}\label{eq:intro-estimator}
\E_z[ f(z)] = \int f(z) \langle X' ,  \omega_z zz^\dagger \rangle dz \,.
\end{equation}
The above is a linear function in the entries of $E$ since $X'$ is linear in $E$, and as long as $f(z)$ is chosen appropriately, it will be a linear function of the matrix $E$.  We can write it as $Y  + cE$ for some constant $c$ and some fixed matrix $Y$ independent of $E$ (under our eventual choice of $f$ and $\{\omega_z zz^\dagger\}$, $Y$ will be a multiple of $I_d$ by symmetry).  Now, we have an estimator, given by $f(z)$, whose mean is $cE$ (as we can simply subtract off $Y$) and whose variance is equal to the average of $\norm{f(z)}_F^2$.  Let
\[
\theta = \int \omega_z \norm{f(z)}_F^2 dz \,. 
\]
The quantity $\theta$ controls the variance of our estimator.  Technically $\theta$ is the average of $\norm{f(z)}_F^2$ when measuring the maximally mixed state but when $E$ is small, it is the same up to constant factors as the average of $\norm{f(z)}_F^2$ when measuring $X'$.  Thus, the sample complexity that we need to estimate $E$ to accuracy $\eps$ in Frobenius norm is $O( \theta /(c^2\eps^2))$ \--- in particular it scales linearly in the variance $\theta$ and inverse quadratically with the ``signal" $c$ in the estimator.
\\\\
\noindent Now it suffices to maximize the ratio $c^2/\theta$.  We can express 
\begin{equation}\label{eq:formula-intro}
\langle X', \omega_z z z^\dagger \rangle = \frac{\omega_z}{d^t} + \frac{\omega_z}{d^{t-1}}\langle E,  G_1(z) \rangle
\end{equation}
where the matrix $G_1(z)$ is constructed as follows (see Section~\ref{sec:lower-bound-tools} for more details): 
\begin{itemize}
    \item Rearrange $z$ into a $\underbrace{d \times \dots \times d}_t$ tensor 
    \item Let $F_1(z), \dots , F_t(z)$ be the $t$ different ways to unfold this tensor into a $d \times d^{t-1}$ matrix by flattening all but one of the modes. 
    \item Set $G_1(z) = F_1(z)F_1(z)^\dagger + \dots + F_t(z) F_t(z)^\dagger $
\end{itemize}

\noindent Roughly, we can think of maximizing $c$ as the same as maximizing $\E_z[\langle E, f(z) \rangle ]$.  We can write
\[
\E_z[\langle E , f(z) \rangle] = \frac{1}{d^{t-1}} \int \omega_z \langle E, f(z) \rangle \langle E, G_1(z) \rangle dz \,.
\]
The above is a linear function in $f(z)$ and roughly, we are trying to maximize it subject to a norm constraint on $\norm{f(z)}_F^2$ so a natural choice for $f(z)$ (that turns out to be optimal) is $f(z) = G_1(z)$.  In this case, we can compute that both $c$ and $\theta$ scale with $\int \omega_z \norm{G_1(z)}_F^2 dz$ and thus the quotient $c^2/\theta$ scales linearly with the integral as well.  We get that
\[
\frac{c^2}{\theta} \sim  \frac{1}{d^2} \int \omega_z \norm{G_1(z)}_F^2 dz 
\]
where the factor of $1/d^2$ comes from the fact that $c$ contains an extra factor of $1/d$.  Thus, to summarize, it suffices to understand the average Frobenius norm of the unfolded matrix $G_1(z)$.

\subsubsection{Bounding the Unfolded Matrix}\label{sec:unfolded-matrix}

The key insight in understanding $\norm{G_1(z)}_F^2$ is that it is actually related to the projection of $z$ onto the different Schur subspaces.  We prove in Lemma~\ref{lem:skewness-bound} that 
\begin{equation}\label{eq:key-eq-intro}
\norm{G_1(z)}_F^2 \leq \norm{z}^2 \sum_{\lambda \vdash t} \norm{\Pi_{\lambda} z}^2 ( \lambda_1^2 + \dots + \lambda_d^2) 
\end{equation}
where the sum is over all partitions $\lambda$ of $t$ and $\Pi_{\lambda}$ denotes the projection onto the $\lambda$-Schur subspace of $\C^{d^t}$ (see Section~\ref{sec:rep-theory} for formal definitions).  The point is that the distribution over partitions $\lambda \vdash t$ induced by the projections $\Pi_{\lambda}$ is exactly the well-studied Schur-Weyl distribution.  Since $\{ \omega_z zz^\dagger\}$ need to form a POVM, the average of $\norm{G_1(z)}_F^2$ is roughly determined by the value of $\lambda_1^2 + \dots + \lambda_d^2$ for a ``typical" partition drawn from the Schur-Weyl distribution.   We prove that this typical value is exactly $\Theta(t^{1.5})$ when $t \leq d^2$ (see Claim~\ref{claim:typical-young-tableaux1} and Claim~\ref{claim:typical-young-tableaux2}).  For our learning algorithm, we show that Keyl's POVM \cite{keyl2006quantum} actually achieves $\norm{G_1(z)}_F^2 \sim t^{1.5}$ for most $z$ in the POVM.  For our lower bound, we formalize the intuition that the quantity $\norm{G_1(z)}_F^2$ precisely controls the ``information" gained by the learner after each measurement.

\subsection{Learning Algorithm}

Section~\ref{sec:power-series} and Section~\ref{sec:unfolded-matrix} already give the blueprint for the learning algorithm when $\rho$ is close to maximally mixed.  In particular, we have a POVM (namely Keyl's POVM) given by $\{\omega_z zz^\dagger \}_{z \in \calZ}$ in $\C^{d^t \times d^t}$ such that $\norm{G_1(z)}_F^2 \sim t^{1.5}$ for most $z$ and thus using the estimator $f(z) = G_1(z)$ in \eqref{eq:intro-estimator}, we can use $d^2/(t^{1.5} \eps^2)$ batches and compute $\wh{\rho}$ such that $\norm{\wh{\rho} - \rho}_F \leq \eps$.

\subsubsection{Rotational Invariance}

We first note that naively, the approximation \eqref{eq:intro-powerseries} seems to require $t \ll 1/(d\norm{E})$.  However, we can exploit the symmetries in Keyl's POVM to deduce that the approximation holds for a much larger range of $t$.  We show that the approximation holds up to $t \sim (1/\norm{E}_F)^c$ for some small constant $c$.  In particular, we obtain nontrivial guarantees even when $\norm{E}  \gg 1/d$ (this is why Theorem~\ref{thm:upper-informal} works up to $t \sim (d/\eps)^c$ instead of just $t \sim (1/\eps)^c$).   The key point is that Keyl's POVM is copywise rotationally invariant, meaning that for a unitary $U \in \C^{d \times d}$, applying $U^{\otimes t} zz^\dagger (U^\dagger)^{\otimes t}$ to all of the elements of the POVM results in the same POVM.  The key lemma is Lemma~\ref{lem:random-rotations}, which shows that for any matrix $E$ and unit vector $v$,
\begin{equation}\label{eq:intro-rotation}
{d^t} \left\lvert \E_{U}[\langle (U E U^{\dagger})^{\otimes t}, vv^\dagger \rangle  ] \right\rvert \leq (4t)^{4t} \norm{E}_F^t
\end{equation}
where the expectation is over Haar random unitaries $U$.  In particular, due to the random rotation on the LHS, the RHS scales with $\norm{E}_F$ instead of $d\norm{E}$.  We can use \eqref{eq:intro-rotation} to bound the contribution of the terms of order-$2$ and higher in the expansion $\rho^{\otimes t} = (I_d/d + E)^{\otimes t}$.  Roughly the point is that for constant $k$, there are $\binom{t}{k}$ terms of order-$k$ and they are bounded by  $\norm{E}_F^k$, so as long as $t \sim \norm{E}_F^c$ for some small constant $c$, the error terms in the truncation in \eqref{eq:intro-powerseries} decay geometrically.

\subsubsection{A Multi-stage Approach via Quantum Splitting}

The algorithm we sketched so far only works when $\rho$ is close to $I_d/d$.  
It remains to show how to make it work for general states.
We do so in two steps.
First, we can bootstrap this to learn $\rho$ whenever $\norm{\rho} \leq 2/d$.
This is because we can first obtain a rough estimate $\wh{\rho}$ and then simulate access to the state $\frac{\rho + (2I_d/d - \wh{\rho})}{2}$.  This state is close to maximally mixed so we can obtain a finer estimate of $\rho - \wh{\rho}$ which allows us to refine our estimate of $\rho$.  Note that the same argument works whenever   $\norm{\rho} \leq C/d$ for some constant $C$.

Next, when $\norm{\rho}$ is large, we rely on a splitting procedure that transforms a state $\rho$ into a state $\Split(\rho)$ such that 
\begin{itemize}
\item $\norm{\Split(\rho)} \leq C/d$ for some constant $C$
\item We can simulate measurement access to copies of $\Split(\rho)$ given measurement access to copies of $\rho$
\end{itemize}  
The way the splitting procedure works is as follows: we first (approximately) diagonalize $\rho$.  Let its eigenvalues be $\alpha_1, \dots , \alpha_d$.  For each $j$, let $b_j$ be the smallest nonnegative integer such that $\alpha_j \leq 2^{b_j}/d$.  Then $\Split(\rho)$ is a $k \times k$ matrix where $k = 2^{b_1} + \dots + 2^{b_d}$ (note that $k \leq 4d$).  For each $j$, we split the $j$th row and column of $\rho$ into $2^{b_j}$ rows and columns and divide the entries evenly between these rows and columns (see Definition~\ref{def:split}).  Note that we lose a factor of $2^{\max(b_1, \dots , b_d)/2}$ in our recovery guarantee due to the splitting procedure.  Thus, we incorporate a multi-scale approach where we consider projecting out eigenvalues and splitting  at different scales.  Due to this, our final recovery guarantee for learning general states is in trace norm instead of Frobenius norm. 
 See Section~\ref{sec:embedding} for more details.


\subsection{Lower Bound}

The lower bound follows the likelihood ratio framework of \cite{chen2023does} and the key new ingredient is the upper bound on the average Frobenius norm of $G_1(z)$ in \eqref{eq:key-eq-intro}.  Recall the high level framework in \cite{chen2023does}.  There is some prior distribution on valid quantum states $\rho$.  The unknown state $\rho_0$ is sampled from this distribution.  After making a sequence of measurements and observing outcomes, say $z_1, \dots , z_m$ where $m = n/t$, we have some posterior distribution over states $\rho$ and learning succeeds only when the posterior concentrates around $\rho_0$.  The posterior measure at some alternative hypothesis $\rho$ satisfies
\begin{equation}
\frac{\mu_{\textsf{post}}(\rho)}{\mu_{\textsf{post}}(\rho_0)} = \prod_{j = 1}^m \frac{z_j^\dagger \rho^{\otimes t} z_j}{z_j^\dagger \rho_0^{\otimes t} z_j} \,.
\end{equation}
Let $\rho$ have $\norm{\rho - \rho_0}_1 \geq C_1\eps$ and $ \norm{\rho - \rho_0} \leq C_2\eps/d$ for some constants $C_1, C_2$.  Then to prove a lower bound, it roughly suffices to show that the above ratio is at least $\exp(-d^2)$ (this is because the volume of the ball of alternative hypotheses is $\exp(d^2)$ times the volume of valid outputs with respect to $\rho_0$).

To lower bound the above ratio, it suffices to show that the update to the ratio in each step is at most $\exp( -\eps^2 t^{1.5}/d )$ on average.  The key lemma that proves this is Lemma~\ref{lem:avg-likelihood-ratio}.  While the full lemma is a bit more complicated and only holds in an average-case sense, for the purposes of this overview, we can think of the lemma as saying that 
\[
\frac{z^\dagger \rho^{\otimes t} z}{z^\dagger \rho_0^{\otimes t} z} \geq \exp\left( -\frac{\norm{G_1(z)}_F^2 \eps^2}{d} \right) \,.
\]  
Now, in Section~\ref{sec:unfolded-matrix}, we argued that for any POVM, $\norm{G_1(z)}_F^2$ can be at most $O(t^{1.5})$ on average so typically, the above ratio will be at least $\exp\left(-\frac{t^{1.5} \eps^2}{d}   \right)$.  Thus, as long as $m \leq d^3/(t^{1.5} \eps^2 )$, which is the same as $n \leq d^3/(\sqrt{t} \eps^2)$, then $\frac{\mu_{\textsf{post}}(\rho)}{\mu_{\textsf{post}}(\rho_0)} \geq \exp(-d^2)$ and this completes the sketch of the lower bound.